\documentclass[preprint,12pt]{elsarticle}
\usepackage{amssymb}
\usepackage{amsmath}
\usepackage[ruled,vlined,linesnumbered]{algorithm2e}
\usepackage{graphicx}
\usepackage{booktabs}
\usepackage[hidelinks]{hyperref}
\usepackage{multirow}
\usepackage{pifont}
\usepackage{float}
\usepackage[section]{placeins}
\usepackage{bm}
\usepackage{amsthm}
\usepackage{makecell}
\newcommand{\DC}{\mathcal{F}_{\mathrm{DC}}}
\newcommand{\norm}[1]{\left\|#1\right\|}
\newcommand{\pA}{p_A}
\newcommand{\Phiinv}{\Phi^{-1}}
\DeclareMathOperator*{\argmax}{arg\,max}
\newtheorem{theorem}{Theorem}
\newtheorem{definition}{Definition}

\theoremstyle{remark}
\newtheorem{remark}{Remark}

\journal{Expert Systems with Applications}

\begin{document}

\begin{frontmatter}

\title{Traffic-Aware Randomized Smoothing for LLM-Based Network Intrusion Detection}

\author[1]{Zhenpeng Li}

\affiliation{organization={Guangzhou Health Science College},
            addressline={GuangYuanZhong Road, 248}, 
            city={Guangzhou},
            postcode={510405}, 
            state={Guangdong Province},
            country={China}}

\begin{abstract}
Large language model (LLM)-based intrusion detection systems (IDS)
are increasingly studied for security monitoring, yet their robustness
against feasible traffic manipulation remains largely empirical.
We present Traffic-Aware Randomized Smoothing (TA-RS), a
classifier-agnostic certified defense that injects Gaussian noise
exclusively into the directly controllable (DC)
subspace---features a remote attacker can modify---during both
fine-tuning and certification, aligning the smoothing distribution
with the attacker-controllable subspace.
We identify a critical prerequisite: applying standard randomized
smoothing to clean-trained LLM-IDS yields weak certified accuracy in
three of four (model, dataset) pairs tested (14--33\%, at or below
random) and only 57\% in the fourth
(43~pp below the noise-augmented result); noise-augmented
fine-tuning recovers to 68--100\% on two of three
benchmark datasets (at $\sigma{=}0.25$).
At the L$_\infty$-equivalent threshold $R_\infty{=}\varepsilon\sqrt{|\DC|}$
($\varepsilon{=}0.05$), TA-RS achieves 55--100\% certified accuracy on
CIC-IDS-2018 and HIKARI-2021, with median certified radii
($\tilde{R}{\approx}0.45$--$0.96$) exceeding $R_\infty$ by
$1.8$--$5{\times}$ (across $\sigma{=}0.25$--$1.00$).
Against a fairly trained iso-trained RS baseline the residual
advantage is dataset-dependent (4--19~pp on CIC-IDS-2018).
The larger gap---up to 72~pp against an isotropic RS baseline that
shares the DC-noise-augmented training recipe---primarily reflects the
training--certification mismatch rather than DC alignment alone:
isotropic test-time noise perturbs uncontrollable features the
attacker cannot exploit, triggering abstention rates up to 68\%.
RT-IoT2022 probes the limits of the method: it fails under the default
fine-tuning recipe but recovers to 76\%/69\% certified accuracy
(LLaMA3-8B/Qwen3-8B) when noise augmentation is increased.
The framework provides a principled basis for certified defenses in
traffic-domain settings beyond the configurations evaluated.
\end{abstract}

\begin{graphicalabstract}
  \includegraphics[width=\textwidth]{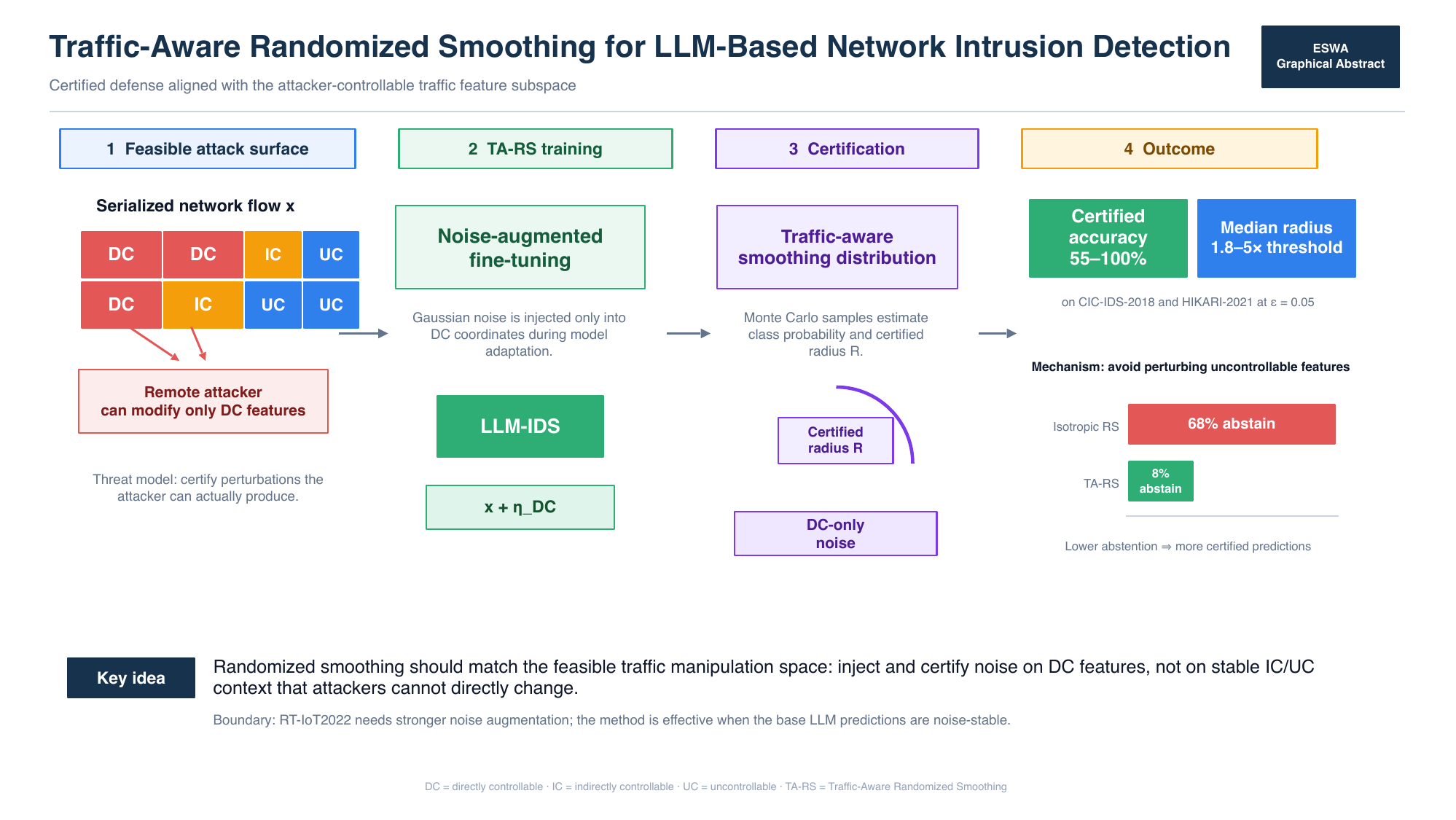}
\end{graphicalabstract}

\begin{highlights}
  \item DC/IC/UC taxonomy identifies which traffic features are attacker-controllable.
  \item Noise-augmented fine-tuning is prerequisite: clean LLMs yield near-random CA.
  \item TA-RS eliminates abstention on non-controllable features (68\%$\to$8\% on CIC).
  \item Anisotropic RS confirms DC-only smoothing is principled ($\leq$2\,pp difference).
  \item Certified accuracy 55--100\% on CIC/HIKARI at L$_\infty$ budget $\varepsilon{=}0.05$.
\end{highlights}


\end{frontmatter}

\section{Introduction}
\label{sec:intro}

Network intrusion detection systems (IDS) increasingly employ
LLM-style models that classify serialized traffic features into
attack categories.  Before serialization, however, a network adversary
can alter packet sizes, flags, or port numbers---features that remain
under attacker control regardless of the classification model.

A companion paper~\cite{companion_a1} formalizes this attack surface by
partitioning network-flow features into \emph{directly controllable}
(DC), \emph{indirectly controllable} (IC), and \emph{uncontrollable}
(UC) sets.  Adversarial perturbations confined to the DC subspace
already exceed 60\% attack success rate against LLM-IDS on multiple
modern datasets.  The defense problem is therefore sharply scoped: can
we certify the classifier against perturbations the attacker can
actually produce?

Empirical defenses---adversarial training~\cite{madry2018towards,goodfellow2015explaining},
preprocessing, anomaly filtering---reduce observed attack success but cannot give a
worst-case guarantee outside the attacks used during evaluation.
Randomized Smoothing (RS)~\cite{cohen2019certified} does provide such a
guarantee, but applying it directly to LLM-IDS misses two conditions
specific to traffic classification.

The first is \emph{threat alignment}.  Standard RS injects isotropic
noise into every feature, including UC features the remote attacker
cannot set.  Certifying those directions inflates the effective
perturbation set beyond what is feasible, and the smoothed classifier
frequently abstains on irrelevant noise directions.

The second---and in our experience the more important---is
\emph{model stability}.  Clean-trained LLM-IDS are not stable under
smoothing noise: in our experiments they reach only 14--57\% certified
accuracy on CIC-IDS-2018 and HIKARI-2021 before any subspace choice
is made.  Without DC-noise-augmented fine-tuning, even perfect
threat alignment cannot produce a useful certificate.

This paper introduces Traffic-Aware Randomized Smoothing (TA-RS) around
these two coupled conditions.  TA-RS smooths only the DC feature
subspace using $\mathcal{N}(0,\sigma^2\mathbf{I}_{\DC})$ and fine-tunes
the classifier with traffic-aware noisy copies so that training and
certification occur in the same attacker-controllable subspace.  The
resulting certificate is a DC-subspace L2 certificate; its relationship
to the L$_\infty$ traffic budgets used in the companion attack model is
made explicit rather than assumed.

TA-RS is classifier-agnostic in principle: it can wrap any differentiable
or black-box classifier.  We instantiate it with LLM-based classifiers
because text-serialized inputs accept heterogeneous feature schemas without
structural modification.  Fixed-feature tabular baselines (e.g., NAXGB)
can equally be wrapped by TA-RS within a dataset, but require
dimension-matched retraining for new datasets.  Section~\ref{sec:cross_domain}
shows a key limitation: dataset-specific fine-tuning is non-negotiable,
and pretraining knowledge does not transfer without traffic-aware tuning.

The experiments test each condition in turn.  Noise-augmented
fine-tuning raises certified accuracy from weak clean-trained baselines
to 68--100\% on CIC-IDS-2018 and HIKARI-2021; with noise-stable models,
traffic-aware smoothing then outperforms isotropic smoothing by up to
72~pp (4--19~pp residual advantage against iso-trained RS), mainly by
eliminating abstention on non-attacker-controllable noise directions.
RT-IoT2022 is included to probe the limits: the default recipe fails
(CA $\leq 24\%$), but increasing $n_\mathrm{aug}$ to 4 recovers
76\%/69\% certified accuracy (LLaMA3/Qwen3), confirming that
noise-stability---not DC-alignment---is the binding constraint in the
hard case.

\noindent\textbf{Contributions.} We make four contributions:
\begin{enumerate}
  \item \textbf{Problem formulation}: we operationalize the DC/IC/UC
    controllability taxonomy as the certification subspace for LLM-based IDS.
    The taxonomy categorizes features by attacker reachability; the
    dataset-specific DC masks and feature counts required to reproduce
    this work are defined in full in Section~\ref{sec:method}
    (Table~\ref{tab:dc_features}).
    Theorem~\ref{thm:ta_rs} applies the
    Cohen et al.~\cite{cohen2019certified} certificate to the DC subspace;
    the contribution is the problem frame, not the certificate formula.
  \item \textbf{Noise-stability requirement}: we demonstrate that clean-trained
    LLM-IDS often yield weak certificates under RS, while traffic-aware
    noise-augmented fine-tuning recovers certified accuracy to
    68--100\% on CIC-IDS-2018 and HIKARI-2021
    (Section~\ref{sec:results}).
  \item \textbf{Mechanism and boundary}: training--certification
    match---not DC selection per se---is the binding factor.  Against a
    fairly iso-trained baseline the residual TA-RS advantage is 4--19~pp;
    matched random-subspace RS performs within 2~pp on average.
    RT-IoT2022 is the failure-boundary case where noise stability, not
    threat alignment, is the prerequisite (Section~\ref{sec:supp_baselines}).
  \item \textbf{Operational L$_\infty$ coverage}: the DC-subspace L2
    certificate covers the L$_\infty$-equivalent threat threshold
    ($R_\infty{=}\varepsilon\sqrt{|\DC|}$, $\varepsilon{=}0.05$) for
    55--100\% of CIC/HIKARI samples, with median certified radii
    $1.8$--$5{\times}$ above $R_\infty$ (Section~\ref{sec:l2linf}).
\end{enumerate}

\section{Related Work}
\label{sec:related}

\subsection{Certified Robustness and the Choice of Perturbation Space}

Formal robustness verification for neural networks began with
satisfiability-based exact verifiers~\cite{katz2017reluplex,
tjeng2019evaluating} and was quickly overtaken by relaxation-based
methods (CROWN~\cite{zhang2018efficient}, $\alpha$-CROWN) capable
of scaling to moderate-depth networks.  These approaches yield exact
or tight certificates but have not scaled to billion-parameter models.
For LLM-IDS, scalability is not the only issue: the certificate must
also be stated over a perturbation set that matches what a network
attacker can change.

Randomized Smoothing~\cite{cohen2019certified} offers a
complementary approach: by evaluating the base classifier under
Gaussian noise and returning the plurality class, one obtains a
smoothed classifier certifiable by a Neyman--Pearson argument.
Lecuyer et al.~\cite{lecuyer2019certified} introduced an earlier
DP-based framework; Cohen et al.\ tightened the bound and proved
its optimality among all noise-based smoothing approaches.
Subsequent work extended RS to other threat models (L1, Lp via
general smoothing distributions~\cite{yang2020randomized}),
improved the base classifier with adversarial training of the
smoothed model~\cite{salman2019provably}, and analyzed failure modes
of RS~\cite{kumar2020certifiedLimits}.
Derandomized Smoothing~\cite{levine2020derandomized} achieves
certificates against patch attacks by averaging over non-overlapping
patches rather than Gaussian noise.
Anisotropic randomized smoothing~\cite{anisoRS2022} optimizes per-feature
noise variances data-drivenly to maximize certified accuracy.
TA-RS fixes $\sigma_{\mathrm{DC}}{=}\sigma$, $\sigma_{\mathrm{IC/UC}}{=}0$
based on threat-model semantics: noise outside $\DC$ is wasted budget that
increases abstention without covering additional attack directions.
Neither this nor prior RS work addresses which traffic features should be
perturbed when the attacker is bounded to a known feature subspace.
TA-RS retains the same certification principle but restricts smoothing
support to the DC subspace.

\subsection{Feasible Adversarial Perturbations in Network IDS}

The vulnerability of ML-based IDS to adversarial examples was
established early~\cite{goodfellow2015explaining, corona2013adversarial, biggio2013evasion}.
Subsequent work demonstrated practical evasion of deep-learning
IDS~\cite{yang2018adversarial, han2021evaluating} and proposed
empirical countermeasures such as adversarial training~\cite{madry2018towards},
ensemble defenses, and input normalization~\cite{apruzzese2023role}.
These defenses are evaluated against
specific attacks; they do not give a worst-case certificate over a
specified traffic-feature set.

Feature-space constraints are central in network security.
Constraint-aware attacks, including the problem-space perspective
of Pierazzi et al.~\cite{pierazzi2020intriguing}, demonstrate that semantic
feasibility changes the effective perturbation space---a certificate for
IDS should therefore be defined over attacker-controllable coordinates,
not the full feature vector.

\subsection{LLM-Based Security Models and Noise Stability}

LLMs have been applied to log analysis~\cite{logllm2024},
vulnerability description~\cite{vuln2024}, and threat
intelligence~\cite{threatintel2024}.  For intrusion detection,
pre-trained language models have also been adapted to network-traffic
classification.  These systems
serialize structured traffic features into text-like inputs, so
small feature perturbations can change the prompt received by the
model.  The question is therefore not only whether an LLM-IDS is
accurate on clean traffic, but whether its predictions remain stable
under the feature perturbations used by a certifier.

Randomized smoothing assumes a stable base classifier; without it,
the smoothed classifier abstains or collapses to a weak majority.
Existing LLM-IDS studies do not establish this condition under
traffic-feature noise.  Our companion work~\cite{companion_a1} provides
the feasible attack model; this paper asks what training and smoothing
distribution are needed to certify an LLM-IDS within it.

\section{Preliminaries}
\label{sec:prelim}

\subsection{LLM-IDS Pipeline}

Let $\mathbf{x} \in \mathbb{R}^d$ denote a normalized network-flow
feature vector.  An LLM-IDS applies a serialization function
$\text{ser}: \mathbb{R}^d \to \Sigma^*$ (mapping feature values to
a natural language prompt) and a fine-tuned LLM $\text{LM}: \Sigma^*
\to \mathcal{Y}$ (mapping prompts to class labels), so
$f(\mathbf{x}) = \text{LM}(\text{ser}(\mathbf{x}))$.

\subsection{Attacker Capability Model}
\label{sec:capability}

Following~\cite{companion_a1}, we partition the $d$ features into
three disjoint sets:
\begin{itemize}
  \item \textbf{Directly controllable (DC)}: $\DC \subset [d]$,
    features the attacker can set precisely (e.g., destination port,
    TCP flags, packet sizes, window size).  $|\DC|$ ranges from
    14 to 27 depending on the dataset (CIC-IDS-2018: 14, HIKARI-2021: 25, RT-IoT2022: 27).
  \item \textbf{Indirectly controllable (IC)}: features that emerge
    from sustained behavioral patterns (e.g., inter-arrival time
    statistics, flow rate aggregates).
  \item \textbf{Uncontrollable (UC)}: features determined by the
    remote host or network infrastructure (e.g., destination byte
    counts, server response metrics).
\end{itemize}

Table~\ref{tab:dc_features} lists the DC features for each dataset evaluated
in this paper, grounded in network-protocol semantics: these are the
forward-direction header and payload fields an attacker sets when crafting
or replaying traffic, independent of what the network or remote host does.
The IC features (flow-level aggregates, inter-arrival statistics) require
sustained behavioral manipulation and are excluded; UC features are fully
determined by the remote host or infrastructure.

Timing features illustrate the classification boundary.
\emph{Forward} inter-arrival time (IAT) fields---the delays between
packets the attacker sends---are DC for HIKARI-2021 and RT-IoT2022,
because the attacker schedules their own outgoing packets
regardless of server behavior.
Bidirectional flow-level IAT aggregates (e.g., \texttt{Flow IAT Mean/Std}
in CIC-IDS-2018) are IC, because they reflect both attacker-side and
server-side packet timing and cannot be set precisely by the attacker alone.

\begin{table}[t]
\centering
\caption{DC feature taxonomy per dataset.  Features are classified by
network-protocol controllability, not model performance.
Binary: TCP flag bits.  Count: integer-valued forward-direction counters.
Continuous: normalized numeric forward-direction fields.}
\label{tab:dc_features}
\setlength{\tabcolsep}{3pt}
\begin{tabular}{lll}
\toprule
Dataset & $|\DC|$ & DC feature groups \\
\midrule
CIC-IDS-2018     & 14 & Dst port, protocol; fwd pkt/byte counts; \\
                 &    & TCP flags (FIN/SYN/RST/PSH/ACK, binary); \\
                 &    & fwd header len; init fwd win bytes; fwd data pkts \\
\midrule
HIKARI-2021      & 25 & Orig/resp port; fwd pkt/data counts; \\
                 &    & fwd header size (tot/min/max); TCP flags (PSH/URG); \\
                 &    & fwd payload stats (min/max/tot/avg/std); \\
                 &    & fwd IAT (min/max/tot/avg/std); \\
                 &    & fwd subflow/bulk bytes \& pkts; fwd win sizes \\
\midrule
RT-IoT2022       & 27 & Orig/resp port, proto, service; \\
                 &    & fwd pkt/data counts; fwd header size (tot/min/max); \\
                 &    & fwd flags (PSH/URG); fwd payload stats; \\
                 &    & fwd IAT stats; fwd subflow/bulk stats; \\
                 &    & fwd init/last window size \\
\bottomrule
\end{tabular}
\end{table}

\begin{remark}[Feature independence]
\label{rem:independence}
The DC-subspace noise model perturbs each DC feature independently.
Features with natural statistical dependencies (e.g., \texttt{fwd\_iat\_avg}
and \texttt{fwd\_iat\_std}) occupy a constrained manifold, so an adversary
restricted to feasible traffic vectors cannot exploit the full L2 ball.
The certificate's soundness is unaffected---independent Gaussian noise
subsumes feasible adversarial perturbations, so the stated radius is a
valid lower bound---but the certificate may be conservative relative to a
tighter analysis of the feasible manifold.
\end{remark}

Let $\mathbf{I}_{\DC} \in \{0,1\}^{d \times d}$ be the diagonal
indicator matrix for DC features ($[\mathbf{I}_{\DC}]_{ii} = 1$ iff
$i \in \DC$, zero otherwise).

\subsection{Threat Model}

An adversary observes a flow $\mathbf{x}$ that $f$ correctly
classifies as malicious and seeks a perturbation $\bm{\delta}^*$
with support in $\DC$ such that $f(\mathbf{x} + \bm{\delta}^*) \neq
f(\mathbf{x})$.  The attacker's budget is parameterized by
$\varepsilon \in \{0.05, 0.15, 0.30\}$ (L$_\infty$ norm over DC
features in normalized space), consistent with~\cite{companion_a1}.
We seek a defense that certifies
$g(\mathbf{x} + \bm{\delta}^*) = f(\mathbf{x})$ for all
$\bm{\delta}^*$ with $\mathrm{supp}(\bm{\delta}^*) \subseteq \DC$
and $\norm{\bm{\delta}^*}_2 \leq R$ for a certifiable radius $R$.
Note that the attacker budget is in L$_\infty$ while the certificate
derived in this paper is in L2; the two are related but not equivalent.
Remark~\ref{rem:l2_linf} and Section~\ref{sec:l2linf} make this gap
explicit, and TA-RS should not be interpreted as a full L$_\infty$
certificate against the companion threat budget.

\subsection{Randomized Smoothing (Cohen et al.\ 2019)}

\begin{definition}[Smoothed classifier]
Given a base classifier $f: \mathbb{R}^d \to \mathcal{Y}$ and
noise distribution $\mathcal{N}(0, \sigma^2 \mathbf{M})$ parameterized by
positive semidefinite $\mathbf{M}$, the smoothed classifier is
\begin{equation}
  g(\mathbf{x}) = \argmax_{c \in \mathcal{Y}}
  \Pr_{\bm{\delta} \sim \mathcal{N}(0,\sigma^2\mathbf{M})}
  \!\bigl[f(\mathbf{x} + \bm{\delta}) = c\bigr].
  \label{eq:smooth}
\end{equation}
\end{definition}

\begin{theorem}[Cohen et al., 2019]\label{thm:cohen}
Let $\mathbf{M} = \mathbf{I}$ (isotropic noise).  Suppose
$g(\mathbf{x}) = c_A$ and let $\pA = \Pr[f(\mathbf{x}+\bm{\delta})
= c_A]$.  If $\pA > \tfrac{1}{2}$, then
$g(\mathbf{x} + \bm{\delta}^*) = c_A$ for all $\norm{\bm{\delta}^*}_2
\leq \sigma\Phiinv(\pA)$.
\end{theorem}

\section{Traffic-Aware Randomized Smoothing}
\label{sec:method}

TA-RS operationalizes the two conditions introduced above.  The
smoothing distribution is aligned with the attacker's feasible feature
subspace, and the base LLM is trained to be stable under that same
distribution.  The first design choice determines what perturbations
the certificate covers; the second determines whether the smoothed
classifier can make a non-abstaining prediction under those
perturbations.

\subsection{Traffic-Aware Noise Distribution}

The threat-alignment component is the smoothing distribution.  Standard
RS uses $\mathbf{M}=\mathbf{I}$ and perturbs every feature.  TA-RS
instead perturbs only DC features:
\begin{equation}
  \bm{\delta} \sim \mathcal{N}\!\bigl(\mathbf{0},\,
  \sigma^2 \mathbf{I}_{\DC}\bigr),
  \label{eq:ta_noise}
\end{equation}
where $\mathbf{I}_{\DC}$ is the DC indicator matrix
(Section~\ref{sec:capability}).  This choice does not change the
randomized-smoothing principle; it changes the support of the noise so
that the smoothed classifier is evaluated in the same feature subspace
that the attacker can directly manipulate.  The empirical claim tested
in Section~\ref{sec:results} is that this alignment reduces abstention
relative to isotropic smoothing.

\textbf{Discrete DC features.}
Some DC features are integer-valued (e.g., packet counts, port numbers).
The implementation clips noisy samples to protocol-specified per-feature bounds:
$x_i' = \mathrm{clip}(x_i + \delta_i,\, [\ell_i, u_i])$, where
$[\ell_i, u_i]$ is the per-feature normalized range derived from
domain-defined limits (e.g.\ non-negative packet counts, binary flag range
$[0,1]$), not from any statistics of the observed data.
Because clip is a deterministic non-expansive map, defining the composed
classifier $h(\mathbf{x}) = f(\mathrm{clip}(\mathbf{x}))$ absorbs this step
into the base classifier; Theorem~\ref{thm:ta_rs} applies to $h$ exactly,
without approximation.
The following characterizes what this means for each discrete feature type.

For \emph{count-valued} DC features (port, packet count, byte count, header
length, window size), normalization places each integer step at
$\Delta_i = 1/(\max_i - \min_i)$, which is $\leq 10^{-4}$ across all
evaluated datasets.  Since $\Delta_i \ll \sigma{=}0.25$, the rounding
error is negligible and the continuous certificate holds to high accuracy.

For \emph{binary flag} DC features (TCP flag bits, normalized to
$\{0,1\}$), rounding constitutes a threshold operation with boundary at
$0.5$.  Crucially, the smallest certified radius used in this paper is
$r{=}0.05$ and the largest is $r{=}0.15$, both strictly below the
rounding boundary $0.5$.  An L2-bounded adversary with budget $R \leq 0.15$
therefore cannot shift any binary flag feature past the rounding threshold
deterministically.  As a result, the continuous Gaussian certificate at
$r \leq 0.15$ is \emph{conservative} for binary flag features: the true
discrete robustness guarantee is at least as strong as the stated L2 bound,
because the adversary's budget does not suffice to flip any flag with
certainty.  Section~\ref{sec:limitations} discusses the residual discretization
gap and a path to a fully rigorous mixed certificate.

For \emph{categorical} DC features (e.g., \texttt{proto} and
\texttt{service} in RT-IoT2022), each category is encoded as a
normalized integer label.  The composed classifier
$h(\mathbf{x}) = f(\mathrm{ser}(\mathrm{clip}(\mathbf{x}+\bm{\delta})))$
absorbs the full preprocessing pipeline---including nearest-integer
recovery---into the base classifier; Theorem~\ref{thm:ta_rs} applies
to $h$ exactly.  The normalized step per category
$\Delta_i = 1/(C_i-1)$ is $0.5$ for \texttt{proto} ($C_i{=}3$)
and ${\approx}0.015$ for \texttt{service} ($C_i{=}66$), placing
both within the conservative regime established for count features.

\subsection{Noise-Augmented Fine-Tuning}
\label{sec:finetune}

The stability component is the training distribution.  A smoothed
classifier is useful only if the base classifier gives consistent
predictions under the same noise used at certification time.  This is
especially important for serialized LLM-IDS: a small change in a
numeric feature can change the prompt tokens and shift the model's
output distribution.  Consistent with image-domain RS practice, where
Gaussian augmentation improves smoothed accuracy~\cite{salman2019provably},
we train the LLM on the same DC-subspace noise used by the certifier.

For each training sample $(\mathbf{x}, y)$, traffic-aware
noise-augmented fine-tuning includes one clean copy and $K{=}2$ noisy
copies $(\mathbf{x}+\bm{\delta}_k,y)$, with
$\bm{\delta}_k \sim \mathcal{N}(0,\sigma_{\mathrm{train}}^2
\mathbf{I}_{\DC})$ as in Eq.~\eqref{eq:ta_noise} and
$\sigma_{\mathrm{train}}=0.25$.  The purpose of this module is not to
change the certified set; that set is fixed by Eq.~\eqref{eq:ta_noise}
and Theorem~\ref{thm:ta_rs}.  Its purpose is to make the base LLM
stable enough that the smoothed classifier can avoid abstention.

\subsection{Certification Procedure}

Given the aligned noise distribution and the noise-augmented base
model, certification follows the standard RS two-stage structure.  For
a test sample $\mathbf{x}$, TA-Certify first identifies the top class
under DC-subspace noise and then estimates a lower confidence bound on
the probability of that class.  Algorithm~\ref{alg:certify} describes
the procedure.

\begin{algorithm}[t]
\caption{Traffic-Aware Certification (TA-Certify)}
\label{alg:certify}
\SetKwInOut{Input}{Input}\SetKwInOut{Output}{Output}
\Input{sample $\mathbf{x}$, base classifier $f$, DC mask $\mathbf{I}_{\DC}$,
noise level $\sigma$, counts $N_0, N$, confidence $\alpha$}
\Output{class prediction $c_A$ and certified radius $R$, or \textsc{Abstain}}

\tcp{Phase 1: identify top class}
Draw $\bm{\delta}_i \sim \mathcal{N}(0, \sigma^2 \mathbf{I}_{\DC})$,
$i=1,\ldots,N_0$\;
$c_A \leftarrow \argmax_c \bigl|\{i : f(\mathbf{x}+\bm{\delta}_i)=c\}\bigr|$\;

\tcp{Phase 2: estimate $\pA$ with confidence}
Draw $\bm{\delta}_i \sim \mathcal{N}(0, \sigma^2 \mathbf{I}_{\DC})$,
$i=1,\ldots,N$\;
$n_A \leftarrow |\{i : f(\mathbf{x}+\bm{\delta}_i) = c_A\}|$\;
$\hat{p}_A \leftarrow \text{CP-Lower}(n_A,\, N,\, \alpha)$
\tcp*{Clopper--Pearson $1{-}\alpha$ lower bound}

\eIf{$\hat{p}_A > \tfrac{1}{2}$}{
  $R \leftarrow \sigma \cdot \Phiinv(\hat{p}_A)$\;
  \Return $(c_A,\, R)$\;
}{
  \Return \textsc{Abstain}\;
}
\end{algorithm}

\noindent We employ $N_0 {=} 20$, $N {=} 200$, $\alpha {=} 0.001$
throughout.  The value $N{=}200$ is smaller than typical image-domain
RS because each certification sample requires LLM inference; the
consequence is a conservative Clopper--Pearson lower bound, analyzed in
Section~\ref{sec:discussion}.

\subsection{Certified Accuracy Metric}

The evaluation metric must penalize both instability and wrong
predictions.  We therefore report \emph{certified accuracy at radius
$r$}:
\begin{equation}
  \mathrm{CA}(r) = \frac{1}{|S|}
  \sum_{\mathbf{x} \in S}
  \mathbf{1}\bigl[g^{\mathrm{TA}}(\mathbf{x}) = y_{\mathbf{x}}
               \;\wedge\; R(\mathbf{x}) \geq r\bigr],
  \label{eq:ca}
\end{equation}
where $S$ is the evaluation set.  Crucially, $\mathrm{CA}(r)$
penalizes both incorrect predictions \emph{and} abstentions,
so a classifier that predicts a majority class under noise cannot obtain
a high score unless it is also correct and certified.

\subsection{Formal Guarantee}

The theorem below states what the preceding design certifies.  It is a
DC-subspace version of the randomized-smoothing guarantee: the support
condition on $\bm{\delta}^*$ is part of the claim, not a post-hoc
interpretation.

\begin{theorem}[Traffic-Aware Certified Robustness]
\label{thm:ta_rs}
Let $g^{\mathrm{TA}}$ be defined by~\eqref{eq:smooth} with
$\mathbf{M} = \mathbf{I}_{\DC}$.  Suppose $g^{\mathrm{TA}}(\mathbf{x})
= c_A$ and $\pA = \Pr_{\bm{\delta}\sim\mathcal{N}(0,\sigma^2
\mathbf{I}_{\DC})}[f(\mathbf{x}+\bm{\delta})=c_A] > \tfrac{1}{2}$.
Then for every $\bm{\delta}^*$ satisfying
$\mathrm{supp}(\bm{\delta}^*) \subseteq \DC$ and
$\norm{\bm{\delta}^*}_2 \leq \sigma\Phiinv(\pA)$,
\begin{equation}
  g^{\mathrm{TA}}(\mathbf{x} + \bm{\delta}^*) = c_A.
\end{equation}
\end{theorem}

\begin{proof}
Since $\bm{\delta}$ and $\bm{\delta}^*$ both have support in $\DC$,
restricting to $\DC$ yields an injection
$\mathbb{R}^{\DC} \hookrightarrow \mathbb{R}^d$ that preserves
inner products and norms.  Within $\DC$, the distribution
$\mathcal{N}(0, \sigma^2 \mathbf{I}_{\DC})$ restricted to $\DC$
is $\mathcal{N}(0, \sigma^2 \mathbf{I}_{|\DC|})$.
The Neyman--Pearson argument of Cohen et al.~\cite{cohen2019certified}
applies in $\mathbb{R}^{|\DC|}$ verbatim, yielding the radius
$R = \sigma\Phiinv(\pA)$ within $\DC$.  Because
$\mathrm{supp}(\bm{\delta}^*) \subseteq \DC$ by assumption,
the certificate transfers to the full feature space.
\end{proof}

\begin{remark}[Subspace scope]
The guarantee is stated over the directly controllable subspace:
it certifies perturbations in $\DC$ with
$\norm{\bm{\delta}^*}_2 \leq R$.  Isotropic RS ($\mathbf{M}=\mathbf{I}$)
produces a certificate of the same form but for the full $\mathbb{R}^d$
space, certifying against perturbations in UC/IC dimensions that the
attacker cannot produce.  This broader certificate can be useful if
the threat model includes those dimensions, but it is mismatched to the
capability-constrained traffic adversary studied here.
\end{remark}

\begin{remark}[L2 certificate vs.\ L$_\infty$ threat budget]
\label{rem:l2_linf}
The companion threat model~\cite{companion_a1} bounds adversarial
perturbations in L$_\infty$: $\norm{\bm{\delta}^*}_\infty \leq
\varepsilon$ over DC features.  Our certificate is in L2: it
guarantees robustness for $\norm{\bm{\delta}^*}_2 \leq R$.  These
norms are related by $\norm{\cdot}_\infty \leq \norm{\cdot}_2
\leq \sqrt{|\DC|}\,\norm{\cdot}_\infty$, so a certificate at
$R{=}0.05$ (L2) covers L$_\infty$ perturbations up to
$R/\sqrt{|\DC|} \approx 0.010$--$0.013$, smaller than $\varepsilon_{\min}
{=}0.05$.  Conversely, the L2 norm of an L$_\infty$ adversarial
example with $\varepsilon{=}0.05$ is at most $0.05\sqrt{|\DC|}
\approx 0.19$--$0.26$; our certified radii ($r{=}0.05$ or $r{=}0.15$)
therefore do not fully cover the companion threat budget.
The second conversion ($R \geq \varepsilon\sqrt{|\DC|}$) is more useful:
it identifies the L2 radius at which a certificate directly covers an
L$_\infty$ adversary; Section~\ref{sec:l2linf} reports what fraction of
CIC/HIKARI samples exceed this threshold.
\end{remark}

\section{Experimental Setup}
\label{sec:setup}

The experiments are organized around the two requirements in the
central claim.  First, the base LLM must make stable predictions under
noise restricted to directly controllable traffic features.  Second,
the certificate should improve because the smoothing distribution
matches the attacker's feasible feature subspace, not because the
baseline is weaker or the evaluation radius is easier.  The setup
therefore keeps the base model, certification budget, sample count,
and noise level matched whenever TA-RS is compared with isotropic RS.

\subsection{Datasets}

We evaluate these requirements on three modern network-traffic
benchmark datasets:

\textbf{CIC-IDS-2018}~\cite{sharafaldin2018cicids}: four-class
classification (DoS, Normal, Exploitation, CredentialAccess) on 2018
network traces; 25 flow-level features after preprocessing, of which
14 are directly controllable (DC).
This is the most label-diverse dataset in our evaluation.

\textbf{HIKARI-2021}: three-class classification
(CredentialAccess, Exploitation, Normal) on encrypted traffic features;
82 features after preprocessing, of which 25 are DC.
Distinctive per-class feature distributions make this a relatively
structured dataset.

\textbf{RT-IoT2022}: five-class classification
(CredentialAccess, DoS, Exploitation, Normal, Probe)
over IoT traffic with heterogeneous device fingerprints and attack
types; 83 features after preprocessing, of which 27 are DC.
High within-class feature variability makes this the most
challenging dataset for certification.

Table~\ref{tab:datasets} summarises the three datasets.
All experiments employ a fixed prompt-aware train/test split to prevent
prompt-template leakage between training and evaluation.
The \texttt{StandardScaler} normalisation is fitted exclusively on the
training split and applied without refitting to the test split.
Feature clipping bounds $[\ell_i, u_i]$ are protocol-derived constants
(e.g.\ packet-count non-negativity, binary flag range $[0,1]$) and are
not estimated from any data split.
The certification subset is drawn once from the held-out test split
using a fixed random seed (42) and was not used for model selection or
hyperparameter tuning.

\begin{table}[t]
\centering
\caption{Dataset summary.
``Cert.~subset'' is the total number of test samples used for
certification (40 per class, stratified).}
\label{tab:datasets}
\setlength{\tabcolsep}{5pt}
\begin{tabular}{lcccc}
\toprule
Dataset & Classes & Features & DC features & Cert.\ subset \\
\midrule
CIC-IDS-2018  & 4 & 25 & 14 & 160 \\
HIKARI-2021   & 3 & 82 & 25 & 120 \\
RT-IoT2022    & 5 & 83 & 27 & 168$^{\dagger}$ \\
\bottomrule
\multicolumn{5}{l}{$^{\dagger}$One minority class contributes 8 samples
  (test-set availability).}
\end{tabular}
\end{table}

We do not evaluate on NSL-KDD or UNSW-NB15 in this work.  NSL-KDD
uses one-hot encoding that inflates the raw 48-feature space to
122 dimensions after encoding; our perturbation budget is defined over
the raw 48-feature space, so DC-subspace certification is not
well-defined in the encoded representation.  UNSW-NB15 raises the same
issue: the 42-feature DC/IC/UC taxonomy is defined at the raw-feature
level, but the data pipeline encodes categorical fields into a
193-dimensional vector, leaving the DC subspace ambiguous.
CICIoT2023~\cite{neto2023ciciot2023} is a natural candidate for future
evaluation; we defer it because it requires extending the DC/IC/UC
taxonomy and re-validating the perturbation budget alignment, which is
orthogonal to the certification framework studied here.
We discuss dataset scope in Section~\ref{sec:limitations}.

\subsection{Base Classifiers}

We evaluate two 8-billion-parameter LLMs fine-tuned for network
traffic classification:
\textbf{Qwen3-8B}~\cite{qwen3} and
\textbf{LLaMA3-8B}~\cite{llama3}.
Both models receive serialized feature prompts produced by a
dataset-specific template that converts normalized flow features
into natural language descriptions.
For noise-augmented fine-tuning, we extend the training set with
$K$ traffic-aware noisy copies per sample at $\sigma{=}0.25$.
The default is $K{=}2$ ($n_\text{aug}{=}2$); the effect of
increasing to $K{=}4$ ($n_\text{aug}{=}4$) is studied as a
sensitivity experiment for RT-IoT2022 (Section~\ref{sec:rt-iot}),
where greater noise exposure is required to stabilize the base model.

\subsection{Baselines}

The baselines separate three claims: whether certification is absent
without smoothing, whether noise-augmented LLMs can be certified, and
whether restricting smoothing to $\DC$ matters once the model has been
trained for noisy inputs.  The main comparison (Table~\ref{tab:main})
uses three configurations:
\begin{enumerate}
  \item \textbf{Undefended}: clean-trained LLM-IDS with no
    smoothing---certified accuracy is zero by definition.
  \item \textbf{Isotropic RS (noise-aug)}: the same DC-noise-augmented
    model wrapped with isotropic Gaussian noise at matching $\sigma$,
    certifying in $\mathbb{R}^d$ rather than $\DC$.
  \item \textbf{TA-RS (ours)}: DC-noise-augmented model with
    traffic-aware Gaussian noise, certifying in $\DC$.
\end{enumerate}

Baseline~(2) uses the same fine-tuned model as TA-RS; the only
test-time difference is the noise support.  This isolates threat
alignment from training, but does not preclude the possibility that
isotropic RS could benefit from a separately tuned training recipe.
Two supplementary controls address this and a second confound:
\begin{enumerate}\setcounter{enumi}{3}
  \item \textbf{Isotropic RS (iso-trained)}: a model fine-tuned with
    isotropic Gaussian noise augmentation and certified with isotropic
    RS.  This controls for training-recipe asymmetry.
  \item \textbf{Random-subspace RS}: the DC-noise-augmented model
    certified with noise restricted to a randomly chosen subspace of
    size $|\DC|$ (seed 99).  This controls for the possibility that
    the TA-RS advantage arises from dimensionality reduction rather
    than from DC-alignment specifically.  The expected overlap between
    a random subspace of size $|\DC|$ and the true DC set is
    $|\DC|^2 / d$ per dataset (CIC: $14^2/25{\approx}7.8$,
    HIKARI: $25^2/82{\approx}7.6$, RT-IoT: $27^2/83{\approx}8.8$
    out of $|\DC|$ features), so the random baseline intentionally
    samples many non-DC features.
\end{enumerate}
Results for baselines~(4) and~(5) are reported in
Section~\ref{sec:supp_baselines}.

\subsection{Certification Parameters and Metrics}

We employ $N_0{=}20$ samples for class selection, $N{=}200$ Monte Carlo
iterations for probability estimation, and $\alpha{=}0.001$ (confidence
$1{-}\alpha = 99.9\%$).  These values determine the conservativeness of
the certificate and are shared across TA-RS and isotropic RS.  Each
certification run evaluates 40 samples per class (stratified), yielding
160 test samples per (model, dataset) configuration for CIC-IDS-2018
(4-class), 120 for HIKARI-2021 (3-class), and 168 for RT-IoT2022
(5-class; one minority class contributes 8 samples due to test-set
availability).  All reported CA values are macro-averaged over classes (not
weighted by natural traffic distribution).

We certify at three noise levels $\sigma \in \{0.25, 0.50, 1.00\}$.
The primary metric is $\mathrm{CA}(r)$ (equation~\eqref{eq:ca}) at
$r \in \{0.05, 0.10, 0.15\}$ (L2 norm in $\DC$), because it counts only
samples that are both correctly classified and certified.  The secondary
metrics diagnose why certification succeeds or fails: abstention rate
$\mathrm{AR}$ measures whether the probability lower bound exceeds the
decision threshold; smooth accuracy
$\mathrm{SA} = \Pr[g^{\mathrm{TA}}(\mathbf{x}) = y]$ measures whether the
smoothed classifier remains correct; and median certified radius
$\tilde{R}$ over non-abstained samples measures the margin of the
certificate.
For aggregate proportions, we report Wilson 95\% confidence intervals
when the interval affects interpretation.

\subsection{Implementation Notes}
\label{sec:impl}

All experiments employ a single NVIDIA RTX PRO 6000 Blackwell (97~GB VRAM)
in BF16 precision; per-(model, dataset) certification across all
$\sigma$ values takes 1--2 hours.
Both LLMs are fine-tuned with LoRA~\cite{hu2022lora} ($r{=}16$, $\alpha{=}32$, dropout
0.05 on q/k/v/o projections), 5 epochs, effective batch 32,
lr $10^{-4}$ cosine with 5\% warmup, seed 42.
Noise augmentation appends $K{=}2$ DC-noisy copies per training sample
($3\times$ dataset), generated with $\mathcal{N}(0,\sigma^2\mathbf{I}_{|\DC|})$
at $\sigma{=}0.25$ clipped to protocol-specified feature bounds.
Qwen3-8B uses ChatML format; LLaMA3-8B uses raw completion;
inference is greedy (\texttt{max\_new\_tokens=16}, batch 64), taking the
first valid class token.
DC masks are derived from the \texttt{feature\_taxonomy} module:
14/25/27 DC features out of 25/82/83 total for
CIC-IDS-2018/HIKARI-2021/RT-IoT2022.

\section{Results}
\label{sec:results}

The results follow the same logic as the method.  Section~\ref{sec:finding1}
tests whether an LLM-IDS can be certified without first learning stable
predictions under DC noise.  Section~\ref{sec:finding2} tests whether,
given a noise-augmented base model, aligning the smoothing distribution
with $\DC$ improves certified accuracy over full-space smoothing.
Sections~\ref{sec:finding3}--\ref{sec:rt-iot} explain the mechanism and
the failure case.  Section~\ref{sec:finding4} checks that the certified
lower bound is consistent with empirical transfer-attack robustness,
without treating the empirical attack as a substitute for certification.

\subsection{Finding 1: Certification Requires Noise-Stable LLM Predictions}
\label{sec:finding1}

Table~\ref{tab:motivation} asks whether threat-aligned smoothing alone
is enough.  The answer is no.  Clean-trained LLM-IDS models degrade
under TA-RS: LLaMA3-8B falls to the random baseline on both datasets
($20.6\%$ on CIC, $33.3\%$ on HIKARI), while Qwen3-8B on HIKARI retains
moderate certified accuracy ($57\%$, still 43~pp below the
noise-augmented result).  On CIC-IDS-2018, both models fall to or below
the 4-class random baseline ($25\%$).  These rows support the first
condition in the thesis: the LLM must be stable under the same DC
feature noise used for certification.

After noise-augmented fine-tuning, $\mathrm{CA}(0.05)$ recovers to
$68$--$100\%$ on these two datasets ($+43$--$55$~pp).  Noise augmentation
is not universally sufficient---RT-IoT2022 fails even with the same recipe---
but it is the necessary prerequisite before threat alignment can matter.

\begin{table}[t]
\centering
\caption{Certified accuracy at $r{=}0.05$ (TA-RS, $\sigma{=}0.25$):
clean-trained vs.\ noise-augmented LLM-IDS.
Random baselines: HIKARI $33\%$ (3-class), CIC $25\%$ (4-class).}
\label{tab:motivation}
\begin{tabular}{llcc}
\toprule
Dataset & Model & Clean-trained & Noise-aug \\
\midrule
CIC-IDS-2018   & LLaMA3-8B & $20.6\%$           & $\mathbf{76\%}$ \\
CIC-IDS-2018   & Qwen3-8B  & $14.4\%$ & $\mathbf{68\%}$ \\
HIKARI-2021    & LLaMA3-8B & $33.3\%$            & $\mathbf{82\%}$ \\
HIKARI-2021    & Qwen3-8B  & $56.7\%$            & $\mathbf{100\%}$ \\
\bottomrule
\end{tabular}
\end{table}

\begin{figure}[!t]
\centering
\includegraphics[width=\columnwidth]{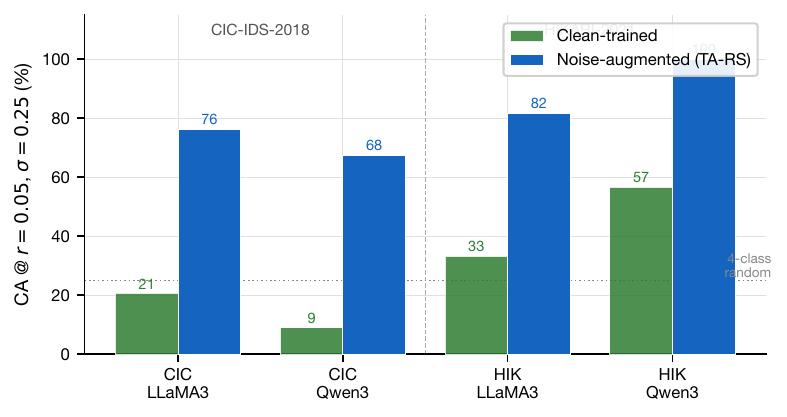}
\caption{Certified accuracy at $r{=}0.05$ (TA-RS, $\sigma{=}0.25$) for clean-trained
vs.\ noise-augmented LLM-IDS on CIC-IDS-2018 and HIKARI-2021.
Clean-trained models often have weak certified accuracy; noise-augmented
fine-tuning recovers certified accuracy to 68--100\%.}
\label{fig:motivation}
\end{figure}

\subsection{Finding 2: Threat-Aligned Smoothing Improves Certification When
the Base Model Is Stable}
\label{sec:finding2}

Table~\ref{tab:main} asks whether the smoothing distribution matters
after the base model has been trained on DC noise.

\noindent\textbf{Security-relevant headline: L$_\infty$-equivalent coverage.}
The primary security metric is certified accuracy at the L$_\infty$-equivalent
threshold $R_\infty = \varepsilon\sqrt{|\DC|}$ for the companion adversarial
budget $\varepsilon{=}0.05$.  At $\sigma{=}0.25$, TA-RS achieves
$69.4\%$/$55.0\%$ on CIC-IDS-2018 (LLaMA3/Qwen3) and $65.0\%$/$100\%$
on HIKARI-2021 at this threshold, while RT-IoT2022 remains low
($3.0\%$/$3.6\%$) under the default recipe.  These are the values
that directly speak to the companion L$_\infty$ threat model;
the r=0.05 columns below demonstrate results at a smaller radius and should
be read as a certified accuracy curve, not the primary security claim.

\noindent\textbf{Smoothing support matters.}
The headline gap at $r{=}0.05$, $\sigma{=}0.25$ is $+72$~pp for
CIC/LLaMA3 (TA-RS $76\%$ vs.\ isotropic $4\%$) and $+48$~pp for
HIKARI/LLaMA3 ($82\%$ vs.\ $33\%$).  Because both methods employ the same
DC-noise-augmented model, the gap is a pure test-time smoothing-support
effect.  The residual advantage against a fairly trained isotropic baseline
is $4$--$19$~pp (Table~\ref{tab:isotrained}).  Median radii
$\tilde{R}{=}0.456$ for CIC/LLaMA3 confirm that certified samples sit
well above the $r{=}0.05$ threshold; increasing $\sigma$ trades CA(0.05)
for larger radii ($\tilde{R}{\approx}0.96$ at $\sigma{=}1.00$).
HIKARI/Qwen3 reaches $100\%$ at all $\sigma$ with zero abstention;
HIKARI/LLaMA3 ($+48$~pp) is the cleaner alignment witness since isotropic
RS is not also saturated.  Wilson 95\% CIs confirm the headline effects
are significant: CIC/LLaMA3 TA-RS $76\%$ [69--82] vs.\ $4\%$ [2--9];
HIKARI/LLaMA3 $82\%$ [74--88] vs.\ $33\%$ [25--42].
RT-IoT2022 is the boundary case: both TA-RS and isotropic RS yield
$4$--$11\%$ CA ($\leq 4$~pp gap), showing that alignment has no value
when the base model is unstable.  Section~\ref{sec:rt-iot} isolates this.

\begin{table*}[t]
\centering
\caption{Certified accuracy (CA, \%) of traffic-aware (TA) and isotropic
(ISO) randomized smoothing in the DC subspace. Within each TA column,
the two values denote CA at $r=0.05/0.15$, respectively.}
\label{tab:main}

\small
\setlength{\tabcolsep}{4pt}
\renewcommand{\arraystretch}{1.12}

\begin{tabular}{llccc@{\hspace{7pt}}ccc@{\hspace{7pt}}cc}
\toprule
& &
\multicolumn{3}{c}{$\sigma=0.25$} &
\multicolumn{3}{c}{$\sigma=0.50$} &
\multicolumn{2}{c}{$\sigma=1.00$} \\
\cmidrule(lr){3-5}
\cmidrule(lr){6-8}
\cmidrule(lr){9-10}

Dataset & Model
& \makecell{TA\\$.05/.15$}
& \makecell{ISO\\$.05$}
& $\Delta$
& \makecell{TA\\$.05/.15$}
& \makecell{ISO\\$.05$}
& $\Delta$
& \makecell{TA\\$.05/.15$}
& SA \\
\midrule

\multirow{2}{*}{CIC-IDS-2018}
& LLaMA3-8B
& $\mathbf{76}/71$ & 4  & $+72$
& $\mathbf{79}/75$ & 18 & $+62$
& $67/65$ & 78 \\

& Qwen3-8B
& $\mathbf{68}/59$ & 22 & $+46$
& $\mathbf{71}/66$ & 22 & $+49$
& $63/59$ & 74 \\

\addlinespace[2pt]

\multirow{2}{*}{HIKARI-2021}
& LLaMA3-8B
& $\mathbf{82}/70$ & 33 & $+48$
& $\mathbf{78}/72$ & 33 & $+45$
& $72/70$ & 87 \\

& Qwen3-8B
& $\mathbf{100}/100$ & 98 & $+2$
& $\mathbf{100}/100$ & 81 & $+19$
& $100/100$ & 100 \\

\addlinespace[2pt]

\multirow{2}{*}{RT-IoT2022}
& LLaMA3-8B
& $4/4$ & $\mathbf{5}$ & $-1$
& $\mathbf{7}/5$ & 5 & $+1$
& $8/6$ & 30 \\

& Qwen3-8B
& $\mathbf{10}/4$ & 9 & $+1$
& $\mathbf{11}/5$ & 7 & $+4$
& $10/8$ & 36 \\

\bottomrule
\end{tabular}

\vspace{2pt}
\begin{minipage}{0.98\textwidth}
\footnotesize
\textit{Notes:}
$\Delta=\mathrm{CA}_{\mathrm{TA}}(0.05)
-\mathrm{CA}_{\mathrm{ISO}}(0.05)$ in percentage points.
ISO uses the same DC-noise-augmented model as TA; results for separately
trained isotropic models are reported in Table~\ref{tab:isotrained}.
Bold indicates the higher CA at $r=0.05$ within each noise setting.
The median certified radii $\tilde{R}$ at $\sigma=0.25$ are
0.456/0.339 for CIC-IDS-2018,
0.456/0.456 for HIKARI-2021, and
0.045/0.097 for RT-IoT2022,
where each pair corresponds to LLaMA3-8B/Qwen3-8B.
The $\ell_\infty$ interpretation follows the conversion in
Section~\ref{sec:l2linf}.
\end{minipage}
\end{table*}

\begin{figure}[!t]
\centering
\includegraphics[width=\columnwidth]{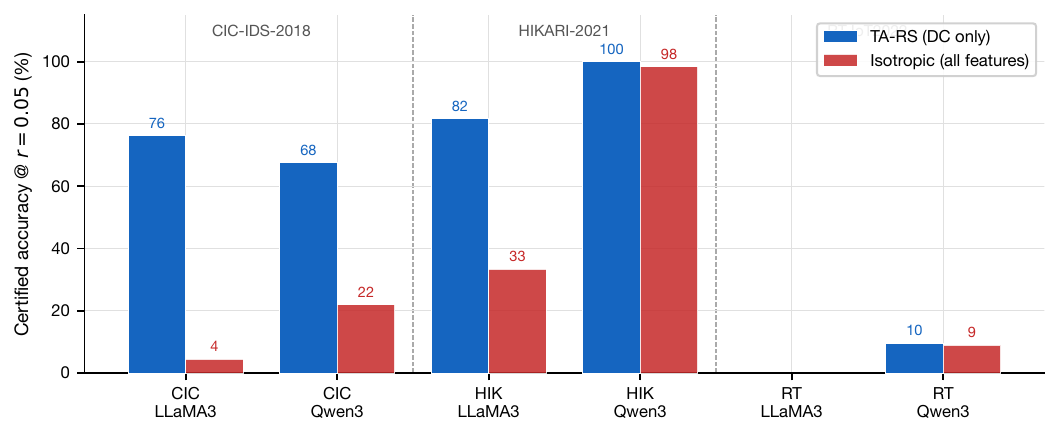}
\caption{Certified accuracy at $r{=}0.05$, $\sigma{=}0.25$.
Both TA-RS and isotropic RS employ the same DC-noise-augmented model; the
72~pp gap reflects training--certification mismatch, not DC alignment alone
(fair comparison: 4--19~pp, Table~\ref{tab:isotrained}).
RT-IoT2022 shows low CA under both methods ($\leq 1$~pp gap).}
\label{fig:traffic_vs_iso}
\end{figure}

\subsection{Finding 3: Abstention Identifies the Mechanism}
\label{sec:finding3}

Table~\ref{tab:abstain} asks why threat-aligned smoothing improves
certified accuracy on CIC-IDS-2018 and HIKARI-2021.  The main mechanism
is abstention.  For CIC-IDS-2018 / LLaMA3-8B, isotropic RS abstains on
$68\%$ of samples, while TA-RS abstains on only $8\%$.  This $60$~pp
abstention gap accounts for most of the $72$~pp certified-accuracy gap.

This pattern is consistent with the proposed mechanism.  When isotropic
noise is injected into UC features (e.g., destination byte counts,
server timing statistics), the base classifier---even after
noise-augmented training on \emph{DC} noise---sees feature distributions
outside the stability region induced by training.  The resulting vote
distribution is often too flat for the Clopper--Pearson lower bound to
exceed 0.5, so the certification algorithm abstains.  TA-RS avoids this
failure mode by perturbing only the DC coordinates.

\begin{table}[t]
\centering
\caption{Abstention rate ($\%$) at $\sigma{=}0.25$.
High isotropic abstention explains the traffic-aware advantage.}
\label{tab:abstain}
\begin{tabular}{llcc}
\toprule
Dataset & Model & TA-RS & Isotropic \\
\midrule
CIC-IDS-2018 & LLaMA3-8B & $8$  & $68$ \\
CIC-IDS-2018 & Qwen3-8B  & $14$ & $44$ \\
HIKARI-2021  & LLaMA3-8B & $3$  & $15$ \\
HIKARI-2021  & Qwen3-8B  & $0$  & $1$  \\
RT-IoT2022   & LLaMA3-8B & $78$ & $76$ \\
RT-IoT2022   & Qwen3-8B  & $60$ & $27$ \\
\bottomrule
\end{tabular}
\end{table}

\begin{figure}[!t]
\centering
\includegraphics[width=\columnwidth]{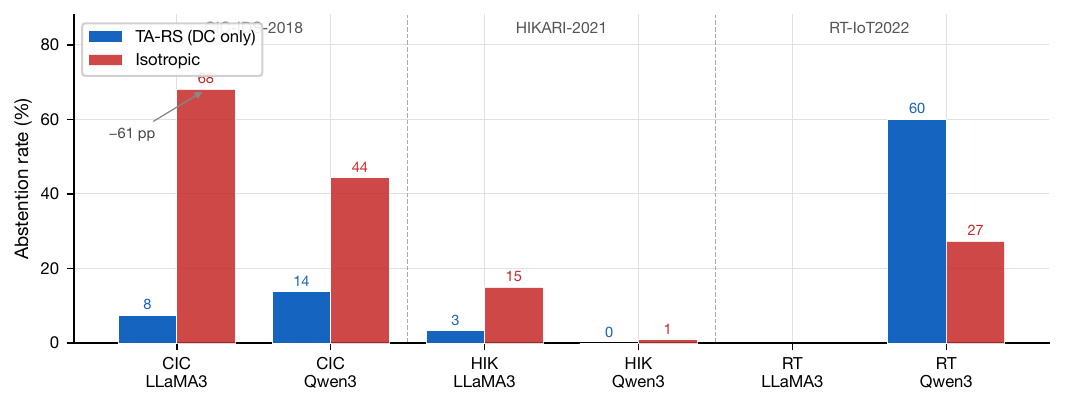}
\caption{Abstention rate (\%) for TA-RS vs.\ isotropic RS ($\sigma{=}0.25$).
High isotropic abstention on CIC-IDS-2018 (44--68\%) vs.\ low TA-RS abstention
(8--14\%) directly explains the certified-accuracy gap.}
\label{fig:abstain}
\end{figure}

Fig.~\ref{fig:abstain} visualizes the same abstention pattern across all
configurations.  The figure supports the mechanism claim for CIC and
HIKARI, but it also shows the limit of the mechanism: on RT-IoT2022,
TA-RS still abstains heavily under the default training recipe.

\noindent\textbf{Smooth accuracy as the stability check.}
Table~\ref{tab:smooth} asks whether the smoothed classifier remains
correct under DC noise before abstention is considered.  On CIC-IDS-2018,
the noise-augmented TA-RS classifiers maintain $74$--$82\%$ smooth
accuracy across the tested $\sigma$ values, compared with the
clean-trained model's $\approx 84\%$ clean accuracy prior to noise
augmentation.  HIKARI / Qwen3-8B retains $100\%$ smooth accuracy at all
$\sigma$ values, which explains why certification is easy for that
pair.  RT-IoT2022 has much lower smooth accuracy ($30$--$42\%$), so its
low certified accuracy is not only an abstention problem; the smoothed
classifier is often incorrect.

\begin{table}[t]
\centering
\caption{Smooth accuracy SA (\%) of the TA-RS smoothed classifier
across noise levels.  SA measures the fraction of all test samples
(including abstentions) correctly classified by the smoothed classifier.}
\label{tab:smooth}
\begin{tabular}{llccc}
\toprule
Dataset & Model & $\sigma{=}0.25$ & $\sigma{=}0.50$ & $\sigma{=}1.00$ \\
\midrule
CIC-IDS-2018 & LLaMA3-8B & 82 & 81 & 78 \\
CIC-IDS-2018 & Qwen3-8B  & 79 & 81 & 74 \\
HIKARI-2021  & LLaMA3-8B & 88 & 87 & 87 \\
HIKARI-2021  & Qwen3-8B  & 100 & 100 & 100 \\
RT-IoT2022   & LLaMA3-8B & 36 & 32 & 30 \\
RT-IoT2022   & Qwen3-8B  & 42 & 38 & 36 \\
\bottomrule
\end{tabular}
\end{table}

\begin{table}[t]
\centering
\caption{$\sigma$ ablation: certified accuracy CA(\%) at $r \in \{0.05, 0.15\}$
and median certified radius $\tilde{R}$ for TA-RS.
Bold indicates the best CA(0.05) per (dataset, model) pair.}
\label{tab:sigma}
\begin{tabular}{llccccc}
\toprule
& & \multicolumn{2}{c}{CA(0.05)} & \multicolumn{2}{c}{CA(0.15)} & $\tilde{R}$ \\
\cmidrule(lr){3-4}\cmidrule(lr){5-6}
Dataset / Model & $\sigma$ & TA & ISO & TA & ISO & (TA) \\
\midrule
CIC / LLaMA3  & 0.25 & 76 & 4  & 71 & 2  & 0.456 \\
              & 0.50 & \textbf{79} & 18 & 75 & 11 & 0.762 \\
              & 1.00 & 67 & 23 & 65 & 21 & 0.958 \\
\midrule
HIKARI / Qwen3 & 0.25 & \textbf{100} & 98 & 100 & 78 & 0.456 \\
               & 0.50 & \textbf{100} & 81 & 100 & 56 & 0.913 \\
               & 1.00 & \textbf{100} & 33 & 100 & 33 & 1.826 \\
\midrule
RT-IoT / LLaMA3 & 0.25 & 4 & 5 & 4 & 3 & 0.045 \\
                & 0.50 & 7 & 5 & 5 & 5 & 0.082 \\
                & 1.00 & \textbf{8} & 6 & 6 & 5 & 0.139 \\
\bottomrule
\end{tabular}
\end{table}

\subsection{Certified Accuracy vs.\ Radius}

\begin{figure}[!t]
\centering
\includegraphics[width=\columnwidth]{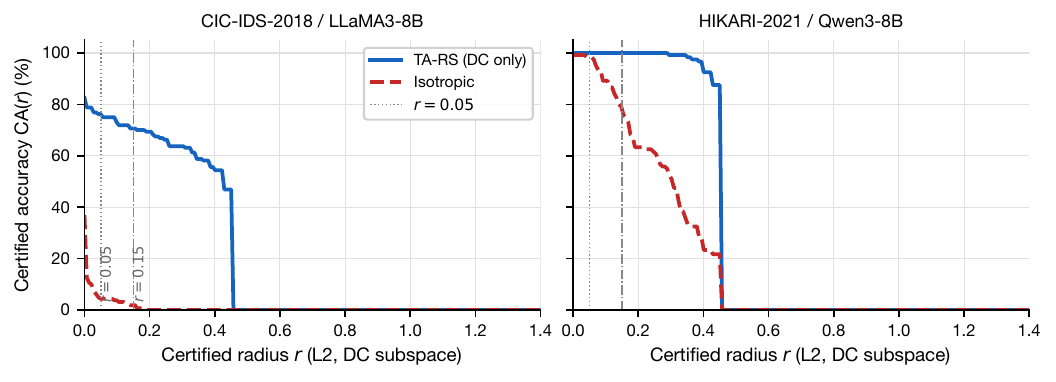}
\caption{$\mathrm{CA}(r)$ vs.\ certification radius $r$ at $\sigma{=}0.25$ for
CIC-IDS-2018/LLaMA3-8B (\emph{left}) and HIKARI-2021/Qwen3-8B (\emph{right}).
TA-RS (solid) degrades gracefully and remains ${\geq}60\%$ past $r{=}0.15$,
while isotropic RS (dashed) drops near zero by $r{=}0.05$.
Vertical lines mark the two reporting thresholds $r \in \{0.05, 0.15\}$.}
\label{fig:certvsradius}
\end{figure}

Fig.~\ref{fig:certvsradius} shows that the $r{=}0.05$ and $r{=}0.15$
operating points sit on a stable CA--radius curve, not isolated peaks.
For CIC/LLaMA3-8B at $\sigma{=}0.25$, CA($r$) stays ${\geq}60\%$
through $r{=}0.30$ and drops sharply only near $r{\approx}0.46$
(the $N{=}200$ ceiling $R_{\max}{\approx}0.453$).
Median radii $\tilde{R} = 0.34$--$0.46$ confirm the certificates have
genuine margin above $r{=}0.05$.
The L$_\infty$-equivalent conversion ($\varepsilon\sqrt{|\DC|}$) is in
Section~\ref{sec:l2linf}.

\subsection{Clean-Robustness Tradeoff}

\begin{figure}[!t]
\centering
\includegraphics[width=\columnwidth]{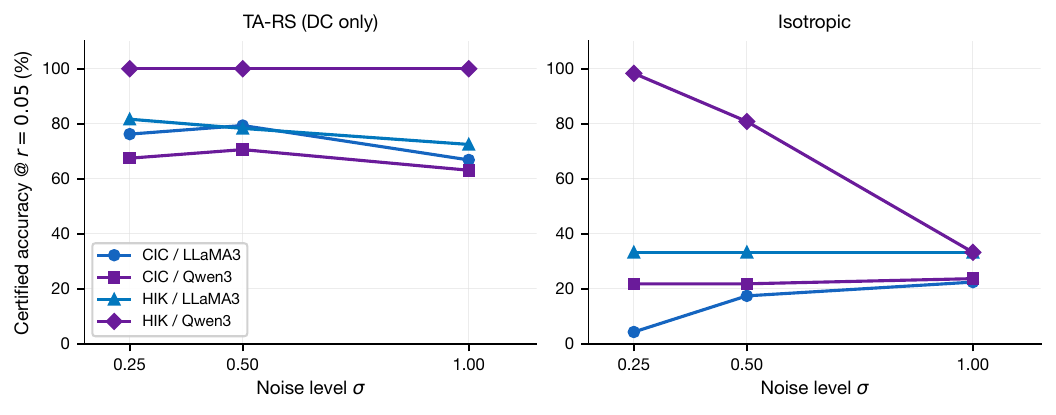}
\caption{Certified accuracy CA(0.05) vs.\ noise level $\sigma \in \{0.25, 0.50, 1.00\}$
for TA-RS (\emph{left}) and isotropic RS (\emph{right}).
CIC-IDS-2018/LLaMA3-8B peaks at $\sigma{=}0.50$ (79\% TA-RS);
HIKARI-2021/Qwen3-8B holds 100\% across all $\sigma$ under TA-RS
but collapses to 33\% under isotropic at $\sigma{=}1.00$.}
\label{fig:tradeoff}
\end{figure}

Fig.~\ref{fig:tradeoff} and Table~\ref{tab:sigma} demonstrate the $\sigma$
tradeoff: increasing $\sigma$ expands $\tilde{R}$ (0.456 $\to$ 0.958
for CIC/LLaMA3) with mild SA drop (82\% $\to$ 78\%), matching the
concave RS frontier~\cite{salman2019provably}.  CA(0.05) peaks at
$\sigma{=}0.50$ (79\%).  For HIKARI/Qwen3, TA-RS holds 100\% across
all $\sigma$ while isotropic RS falls from 98\% to 33\%---evidence
that full-space noise harms UC-stable datasets.  On RT-IoT2022,
$\sigma$ changes add only 4~pp; Section~\ref{sec:rt-iot} isolates the
cause.

\subsection{RT-IoT2022: Sensitivity to Noise-Augmented Training}
\label{sec:rt-iot}

\noindent\textbf{Baseline failure.}
RT-IoT2022 tests the second condition in the thesis by removing it.  With
the default noise-copy count ($n_\text{aug}{=}2$, same as CIC/HIKARI),
TA-RS achieves CA(0.05)~$\leq 10\%$ because the LLM does not learn stable
predictions under DC noise.  For LLaMA3-8B: SA~$=$~$36\%$ (barely above the
$20\%$ five-class random baseline) and abstention~$=$~$78\%$.  For Qwen3-8B:
SA~$=$~$42\%$ and abstention~$=$~$60\%$.  The TA-vs-ISO gap collapses to
${\leq}1$~pp for both models, indicating that restricting noise to $\DC$
cannot rescue a smoothed classifier whose votes are mostly wrong or
inconclusive.

\noindent\textbf{Ablation: noise copy count.}
We retrain both LLaMA3-8B and Qwen3-8B on all three datasets with
$n_\text{aug}{=}4$ (doubling the number of noise-augmented copies per
training sample) and re-certify at $\sigma{=}0.25$.
Table~\ref{tab:aug_ablation} reports SA, abstention, and CA(0.05) for
both $n_\text{aug}$ values across all six (model, dataset) configurations.
For RT-IoT2022 with $n_\text{aug}{=}4$:
\begin{itemize}
  \item SA recovers from $36\%$/$42\%$ to $\mathbf{83\%}$/$\mathbf{70\%}$
    (LLaMA3/Qwen3), reaching CIC/HIKARI levels;
  \item abstention drops from $78\%$/$60\%$ to $\mathbf{12\%}$/$\mathbf{3\%}$;
  \item CA(0.05) increases from $4\%$/$10\%$ (Table~\ref{tab:main}) to $\mathbf{76\%}$/$\mathbf{69\%}$.
\end{itemize}

\begin{table}[t]
\centering
\caption{Sensitivity to noise-copy count $n_\text{aug}$ across all six
(model, dataset) configurations ($\sigma{=}0.25$, traffic-aware RS).
CA(0.05) $=$ fraction \emph{correctly classified and certified} at radius
$r{=}0.05$ following~\cite{cohen2019certified}.
SA and Abstain are from certification runs conducted for this ablation.
CA(0.05) for $n_\text{aug}{=}2$ ($^\dagger$) is taken from
Table~\ref{tab:main} for consistency; CA for $n_\text{aug}{=}4$ is from
an independent certification run on the newly trained aug-4 models.
Bold marks CA that improves over the Table~\ref{tab:main} baseline.}
\label{tab:aug_ablation}
\setlength{\tabcolsep}{4pt}
\begin{tabular}{llccc}
\toprule
Dataset & Model & $n_\text{aug}$ & SA / Abstain & CA(0.05) \\
\midrule
\multirow{2}{*}{CIC-IDS-2018} & \multirow{2}{*}{LLaMA3-8B}
  & 2$^\dagger$ & 82\% / 8\%  & 76\% \\
  & & 4           & 74\% / 16\% & 61\% \\[2pt]
\multirow{2}{*}{CIC-IDS-2018} & \multirow{2}{*}{Qwen3-8B}
  & 2$^\dagger$ & 79\% / 14\% & 68\% \\
  & & 4           & 79\% / 6\%  & \textbf{74\%} \\[2pt]
\multirow{2}{*}{HIKARI-2021} & \multirow{2}{*}{LLaMA3-8B}
  & 2$^\dagger$ & 88\% / 3\%  & 82\% \\
  & & 4           & 96\% / 4\%  & \textbf{88\%} \\[2pt]
\multirow{2}{*}{HIKARI-2021} & \multirow{2}{*}{Qwen3-8B}
  & 2$^\dagger$ & 100\% / 0\% & 100\% \\
  & & 4           & 100\% / 0\% & \textbf{100\%} \\[2pt]
\multirow{2}{*}{RT-IoT2022} & \multirow{2}{*}{LLaMA3-8B}
  & 2$^\dagger$ & 36\% / 78\% & 4\% \\
  & & 4           & 83\% / 12\% & \textbf{76\%} \\[2pt]
\multirow{2}{*}{RT-IoT2022} & \multirow{2}{*}{Qwen3-8B}
  & 2$^\dagger$ & 42\% / 60\% & 10\% \\
  & & 4           & 70\% / 3\%  & \textbf{69\%} \\
\bottomrule
\end{tabular}
\end{table}

\noindent\textbf{Interpretation.}
The ablation isolates noise-stability as the binding constraint.
With $n_\text{aug}{=}4$, CA rises by $+72$~pp for LLaMA3
($4\% \to 76\%$) and $+59$~pp for Qwen3 ($10\% \to 69\%$), while
the TA-vs-ISO gap remains large ($+52$/$+45$~pp)---DC-alignment
continues to matter once the base model is stable.
The isotropic aug-4 baseline reaches only SA~$=24\%$ on RT-IoT2022
(barely above the 20\% random floor), confirming that full-feature
noise augmentation alone cannot substitute for DC-alignment.

A sample-level $\pA$ diagnostic makes the mechanism concrete.
Under $n_\text{aug}{=}2$, $78.0\%$ of RT-IoT/LLaMA3 samples have
$\pA \leq 0.5$ (mean $\pA{=}0.40$)---the smoothed classifier simply
does not reach a decisive majority vote.  Under $n_\text{aug}{=}4$,
only $12.5\%$ abstain and $76.2\%$ reach $\pA{>}0.8$.
By contrast, CIC and HIKARI have abstention rates of $7.5\%$ and
$3.3\%$ under the default recipe (mean $\pA{\approx}0.87$ and $0.86$),
placing the RT-IoT2022 failure squarely in dataset-specific noise
instability, not a certificate design flaw.

One boundary condition: CIC/LLaMA3 CA \emph{decreases} at
$n_\text{aug}{=}4$ ($76\% \to 61\%$) as abstention rises from $8\%$
to $16\%$, signaling over-regularization.  The aug-4 result is
therefore diagnostic rather than prescriptive; we do not replace
Table~\ref{tab:main} with it.

Fig.~\ref{fig:bydataset} illustrates the three-way contrast
(SA, abstention, CA) for the $n_\text{aug}{=}2$ configurations and the
RT-IoT2022 noise-copy ablation.

\begin{figure}[!t]
\centering
\includegraphics[width=\columnwidth]{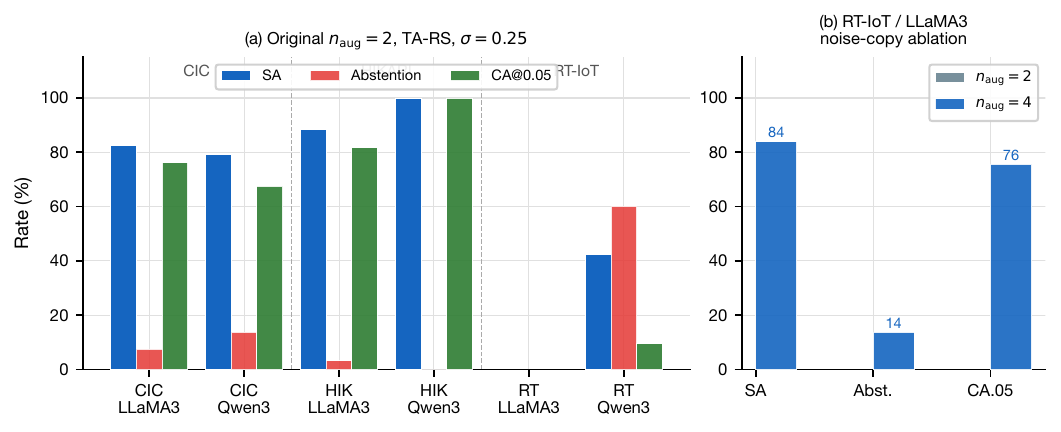}
\caption{(a) SA, abstention rate, and CA@0.05 for TA-RS ($\sigma{=}0.25$,
$n_\text{aug}{=}2$) across all six (model, dataset) configurations.
(b) Noise-copy ablation on RT-IoT2022 ($\sigma{=}0.25$): increasing
$n_\text{aug}$ from 2 to 4 raises SA from $36\%$ to $83\%$ (LLaMA3) and
from $42\%$ to $70\%$ (Qwen3), dropping abstention to $12\%$/$3\%$, and
lifting CA(0.05) from $4\%$/$10\%$ to $76\%$/$69\%$ respectively.}
\label{fig:bydataset}
\end{figure}

\subsection{Fair Baseline Comparisons}
\label{sec:supp_baselines}

The main-table baseline (isotropic RS with DC-trained model) shares the
training recipe with TA-RS, isolating the test-time noise-support choice.
This section reports two additional baselines needed to rule out confounds:
(4)~isotropic RS with iso-trained models (training-recipe fairness check)
and (5)~random-subspace RS (dimensionality-reduction control).
Both are required for a complete assessment of whether TA-RS advantages
stem from DC-alignment or from some other factor.

\noindent\textbf{Exp~A: Random-Subspace RS.}
To test whether the TA-RS advantage arises purely from using a
smaller-dimensional noise subspace (rather than from alignment with
attacker-controllable features), we certify the same DC-noise-augmented
models with noise restricted to a \emph{randomly chosen} subspace of
size $|\DC|$.  If dimensionality reduction alone explained the advantage,
random-subspace RS should perform comparably to TA-RS.

Table~\ref{tab:randsubspace} shows CA(0.05) at $\sigma{=}0.25$ for TA-RS,
random-subspace RS, and isotropic RS (unfair baseline using the same
DC-trained model), where CA~$=$~fraction correctly classified AND certified.
On CIC-IDS-2018 and HIKARI-2021, TA-RS exceeds
random-subspace RS by $22$--$55$~pp, confirming that DC-alignment
contributes beyond mere dimensionality reduction.  Two patterns in
Table~\ref{tab:randsubspace} deserve closer inspection.

First, HIKARI-2021 / Qwen3-8B shows near-identical performance
(TA-RS 100\% vs.\ RandSub 99\%), which might seem to weaken the
DC-alignment argument.  HIKARI has $|\DC|{=}25$ out of $d{=}82$
total features, so a randomly chosen subspace of size 25 has an
expected $|\DC|^2/d \approx 7.6$ features in common with the true
DC set---only 30\% of the DC features by expectation.  The near-tie
for Qwen3-8B therefore reflects near-ceiling noise stability of this
model on HIKARI: when a model is stable on almost all features, any
size-25 subspace tends to yield high CA.  The LLaMA3-8B row
(TA-RS 82\% vs.\ RandSub 60\%), which is not at ceiling, more
directly demonstrates the 22~pp DC-alignment advantage.

Second, on RT-IoT2022 (both models), random-subspace RS substantially
outperforms TA-RS ($+47$--$66$~pp).
This result is a controlled indicator of DC-specific instability, not
a general failure of subspace smoothing.
RT-IoT2022 has 27 DC features out of 83 total; a random subspace of the
same size has an expected overlap of only $27^2/83 \approx 8.8$ DC features,
drawing mostly from IC and UC coordinates for which the model exhibits
greater noise stability under default training.
The DC subspace is thus the \emph{hardest} subspace to be noise-stable on
for this dataset---precisely the features the attacker controls are the ones
the default fine-tuning recipe has not made robust.  This confirms that
the RT-IoT2022 failure is specifically a DC-feature instability, and motivates
the $n_\text{aug}$ ablation (Section~\ref{sec:rt-iot}) as a targeted remedy.
To verify robustness to mask choice, we repeat Exp~A with 10 independently
drawn masks (seeds 99, 101, 103, 107, 109, 113, 131, 137, 139, 149) and
report mean CA and 95\% CI (Table~\ref{tab:randsubspace}).
TA-RS leads RandSub on CIC-IDS-2018 by $26\pm19$~pp (LLaMA3) and
$30\pm9$~pp (Qwen3); the LLaMA3-8B HIKARI gap is $26\pm11$~pp.
RT-IoT2022 confirms the DC-instability pattern across all 10 seeds,
consistent with the noise-stability interpretation.

\begin{table}[t]
\centering
\caption{Exp~A: CA(0.05) at $\sigma{=}0.25$ for TA-RS, random-subspace RS
($^\dagger$single representative mask, seed~99; see text for 10-seed
analysis), and isotropic RS (same DC-trained model).
CA~$=$~fraction correctly classified AND certified ($r{=}0.05$).
Bold marks the highest CA per row.  RT-IoT2022 random-subspace performance
reflects base-model instability under any same-size noise subspace
(Section~\ref{sec:rt-iot}).}
\label{tab:randsubspace}
\setlength{\tabcolsep}{4pt}
\begin{tabular}{llccc}
\toprule
Dataset & Model & TA-RS & RandSub$^\dagger$ & Iso \\
\midrule
CIC-IDS-2018 & LLaMA3-8B & \textbf{76\%} & 21\% & 4\%  \\
CIC-IDS-2018 & Qwen3-8B  & \textbf{68\%} & 29\% & 22\% \\
HIKARI-2021  & LLaMA3-8B & \textbf{82\%} & 60\% & 33\% \\
HIKARI-2021  & Qwen3-8B  & \textbf{100\%}& 99\% & 98\% \\
RT-IoT2022   & LLaMA3-8B & 4\%  & \textbf{70\%} & 5\%  \\
RT-IoT2022   & Qwen3-8B  & 10\% & \textbf{57\%} & 9\%  \\
\bottomrule
\end{tabular}
\end{table}

\noindent\textbf{Exp~B: Iso-Trained Baseline.}
Baseline~(2) in the main comparison uses a DC-noise-augmented model with
isotropic certification---a training-recipe mismatch that may understate
the true isotropic RS performance.  To correct for this, we fine-tune
separate models using isotropic Gaussian noise augmentation
(iso-trained) and certify them with isotropic RS.

Table~\ref{tab:isotrained} quantifies the fair-baseline gap for matched
$n_\text{aug}$.  Under $n_\text{aug}{=}2$: TA-RS leads iso-trained RS
on CIC ($+19$~pp/$+4$~pp for LLaMA3/Qwen3).  On HIKARI, iso-trained
reaches 99\% (near ceiling), ahead of TA-RS for LLaMA3 but matching
Qwen3.  On RT-IoT2022, iso-trained wins decisively
($67\%$/$60\%$ vs.\ $4\%$/$10\%$)---because full-feature augmentation
incidentally stabilises the base model, not because isotropic noise is
traffic-aligned.

To ensure that the RT-IoT2022 comparison at $n_\text{aug}{=}4$ is also
fair, we train iso-trained models with $n_\text{aug}{=}4$ for all six
configurations and certify them.  The lower block of
Table~\ref{tab:isotrained} reports the resulting CA.  At $n_\text{aug}{=}4$,
the iso-trained models also improve substantially on RT-IoT2022 (reflecting
that more augmentation helps any training recipe), and the table reveals
whether the residual TA-RS vs.\ iso-trained gap persists once augmentation
budget is equalised.

\begin{table}[t]
\centering
\caption{Exp~B: CA(0.05) at $\sigma{=}0.25$.
Upper block ($n_\text{aug}{=}2$): TA-RS vs.\ iso-trained RS (fair baseline)
vs.\ iso RS on DC-trained model (unfair baseline from Table~\ref{tab:main}).
Lower block ($n_\text{aug}{=}4$): same comparison with matched augmentation
budget, isolating the effect of smoothing-subspace alignment from stability.
CA~$=$~fraction correctly classified AND certified ($r{=}0.05$).
Bold marks the higher CA between TA-RS and iso-trained within each block.}
\label{tab:isotrained}
\setlength{\tabcolsep}{4pt}
\begin{tabular}{llccc}
\toprule
Dataset & Model & TA-RS & Iso-Trained & Iso (unfair) \\
\midrule
\multicolumn{5}{l}{\textit{$n_\text{aug}=2$ (matched augmentation)}} \\
CIC-IDS-2018 & LLaMA3-8B & \textbf{76\%} & 57\% & 4\%  \\
CIC-IDS-2018 & Qwen3-8B  & \textbf{68\%} & 64\% & 22\% \\
HIKARI-2021  & LLaMA3-8B & 82\% & \textbf{99\%} & 33\% \\
HIKARI-2021  & Qwen3-8B  & \textbf{100\%}& \textbf{99\%}& 98\% \\
RT-IoT2022   & LLaMA3-8B & 4\%  & \textbf{67\%} & 5\%  \\
RT-IoT2022   & Qwen3-8B  & 10\% & \textbf{60\%} & 9\%  \\[2pt]
\multicolumn{5}{l}{\textit{$n_\text{aug}=4$ (matched augmentation)}} \\
CIC-IDS-2018 & LLaMA3-8B & \textbf{61\%} & 45\% & --- \\
CIC-IDS-2018 & Qwen3-8B  & \textbf{74\%} & 68\% & --- \\
HIKARI-2021  & LLaMA3-8B & 88\% & \textbf{99\%} & --- \\
HIKARI-2021  & Qwen3-8B  & \textbf{100\%}& 99\% & --- \\
RT-IoT2022   & LLaMA3-8B & \textbf{76\%} & 58\% & --- \\
RT-IoT2022   & Qwen3-8B  & \textbf{69\%} & 52\% & --- \\
\bottomrule
\end{tabular}
\end{table}

The two certificates differ in scope: iso-trained RS certifies all
$\mathbb{R}^d$ directions (including IC/UC the attacker cannot modify),
while TA-RS is scoped to the DC subspace---the guarantee the traffic
threat model actually requires.  On CIC/LLaMA3, DC-alignment cuts
abstention by $19$~pp, raising CA even though the iso-trained certificate
covers a broader perturbation set.  The RT-IoT2022 reversal is a
base-model stability failure, not a DC-alignment failure: when noise
stability is insufficient, more augmentation (Section~\ref{sec:rt-iot})
is the primary lever.

\textbf{Exp~C: Matched random-subspace baseline.}
Table~\ref{tab:isotrained} compares TA-RS against an iso-trained model
certified with DC-subspace noise---but DC-subspace certification is itself
a deliberate design choice that could be replaced by a randomly chosen
$K$-dimensional subspace of the same size ($K{=}|\DC|$).
To isolate the effect of \emph{subspace selection} from that of
\emph{training-certification matching}, we certify the iso-trained models
(same checkpoints as Table~\ref{tab:isotrained}) with $10$ independently
sampled random subspaces of size $K{=}|\DC|$ each.
Across six configurations ($3$ datasets $\times$ $2$ models, $10$ seeds each),
the outcome is mixed: DC-subspace certification wins decisively on CIC/Qwen3
($\Delta{=}{-13}$~pp vs.\ random mean) and is within $1$~pp on HIKARI-2021
(both methods saturate near $100\%$); random subspace wins on
RT-IoT2022/LLaMA3 ($\Delta{=}{+20}$~pp, $|\DC|{=}27$) and CIC/LLaMA3
($\Delta{=}{+6}$~pp).  Averaged over all six configurations, the mean CA
is $70.0\%$ for DC-cert and $72.0\%$ for random-cert, a $2$~pp gap in
favour of random---within the seed-to-seed variability of random-subspace
itself (std $1$--$9$~pp).  The results confirm that neither method dominates
numerically; the case for DC selection rests on \emph{semantic alignment}
(the DC subspace corresponds to traffic features the adversary actually
controls) and \emph{determinism} (a fixed, reproducible guarantee rather
than a seed-dependent one).

\noindent\textbf{Exp~D: Anisotropic RS.}
Anisotropic Randomized Smoothing~\cite{anisoRS2022} generalises the
noise distribution by assigning a per-feature variance, potentially improving
over the binary $\sigma_{\mathrm{IC}}{=}0$ choice of TA-RS.
To quantify this, we certify the same DC-noise-augmented models under
anisotropic noise ($\sigma_{\mathrm{DC}}{=}0.25$, $\sigma_{\mathrm{IC/UC}}{=}0.05$)
while retaining the DC-subspace certificate formula.
Table~\ref{tab:aniso} reports CA(0.05) at $\sigma{=}0.25$.

\begin{table}[t]
\centering
\caption{Exp~D: CA(0.05) at $\sigma{=}0.25$ for TA-RS vs.\ anisotropic RS
($\sigma_{\mathrm{IC/UC}}{=}0.05$) using the same DC-noise-augmented models.
Both employ the DC-subspace certificate; the only difference is the
inference-time noise on IC/UC features.}
\label{tab:aniso}
\setlength{\tabcolsep}{4pt}
\begin{tabular}{llcc}
\toprule
Dataset & Model & TA-RS & Aniso-RS \\
\midrule
CIC-IDS-2018 & LLaMA3-8B & 76\% & 78\% \\
CIC-IDS-2018 & Qwen3-8B  & 68\% & 67\% \\
HIKARI-2021  & LLaMA3-8B & 82\% & 83\% \\
HIKARI-2021  & Qwen3-8B  & 100\% & 100\% \\
RT-IoT2022   & LLaMA3-8B & 76\% & 74\% \\
RT-IoT2022   & Qwen3-8B  & 69\% & 69\% \\
\bottomrule
\end{tabular}
\end{table}

Aniso-RS and TA-RS differ by at most 2~pp across all six configurations
(absolute differences: $\{+2, -1, +2, 0, -2, 0\}$~pp), well within the
Monte Carlo sampling error at $N{=}200$ (one misclassified sample
$\approx 2.5$~pp).  This confirms two things: (\emph{i})~DC-noise-augmented
training produces a model robust enough that small IC/UC perturbations
at certification time do not materially change pA; and (\emph{ii})~the
design choice $\sigma_{\mathrm{IC}}{=}0$ in TA-RS is principled
(no wasted noise budget on non-attacker-controllable features) rather than
a sensitive hyperparameter---a data-driven anisotropic baseline with
$\sigma_{\mathrm{IC}}{=}0.05$ achieves no advantage over our semantically
motivated choice.

\subsection{Finding 4: Empirical Robustness Is Consistent with the Certified Bound}
\label{sec:finding4}

A core soundness requirement of randomized smoothing is that the
certified accuracy should behave as a conservative lower bound under
fixed empirical attacks.  Table~\ref{tab:empirical} asks whether the
reported certificates are consistent with A1-style transfer attacks
(PGD~\cite{madry2018towards,carlini2017towards} on an XGBoost surrogate; \cite{companion_a1}) against the
noise-augmented TA-RS classifier.  This experiment is a consistency
check, not a replacement for the formal guarantee.

\textbf{Protocol.}  We sample $N_{\text{smooth}}{=}20$ DC-noise
copies per test point ($\sigma{=}0.25$) and take the majority vote
as the smoothed prediction.  Transfer attacks are generated at
$\varepsilon \in \{0.05, 0.15\}$ (L$_\infty$, DC features).
We report attack success rate (ASR) conditional on clean-correct
samples, and compute empirical robustness as $1{-}\text{ASR}_{\text{smoothed}}$.

\begin{table}[!t]
\caption{Empirical robustness vs.\ L2 certified accuracy
($\sigma{=}0.25$, transfer attack).  The attack budget $\varepsilon$
is L$_\infty$, while CA$(r)$ is reported at the same numeric L2 radius;
therefore this table is a consistency check, not a proof that the L2
certificate covers the full L$_\infty$ attack set.}
\label{tab:empirical}
\centering
\scriptsize
\setlength{\tabcolsep}{2.5pt}
\begin{tabular}{llccccc}
\toprule
Dataset & $\varepsilon$ & ASR$_{\text{base}}$ & ASR$_{\text{smooth}}$ & Gain & Emp.\,Rob. & CA$(r{=}\varepsilon)$ \\
\midrule
CIC / Qwen3   & 0.05 & 16.0\% & 8.6\%  & $+$7.4pp  & 91.4\% & 67.5\% $\checkmark$ \\
CIC / Qwen3   & 0.15 & 19.1\% & 8.6\%  & $+$10.5pp & 91.4\% & 59.4\% $\checkmark$ \\
HIKARI / LLaMA3 & 0.05 & 3.9\%  & 1.6\%  & $+$2.4pp  & 98.4\% & 81.7\% $\checkmark$ \\
HIKARI / LLaMA3 & 0.15 & 4.7\%  & 0.0\%  & $+$4.7pp  & 100.0\%& 70.0\% $\checkmark$ \\
\bottomrule
\end{tabular}
\end{table}

Three observations follow from Table~\ref{tab:empirical}.  First,
empirical robustness exceeds the reported L2 certified accuracy in every
row.  Note that these metrics have different denominators: empirical
robustness is conditioned on clean-correct predictions, while
CA$(r)$ is unconditional (penalising both misclassification and abstention).
The observed ordering is therefore partially structural and does not by
itself validate certificate soundness; it is consistent with---but not
proof of---conservativeness.  Additionally, the L$_\infty$/L2 norm
mismatch means the two columns represent different threat sets.  Second, the smoothed classifier
reduces ASR beyond the base model in every measured case: the average
reduction is $7.5$\,pp on CIC and $3.6$\,pp on HIKARI.  Third, the
margin between empirical robustness and certified accuracy is large
(up to $30$\,pp on HIKARI), indicating that the Clopper--Pearson bound
with $N{=}200$ is conservative.  Increasing $N$ would tighten this gap
without changing the empirical defense quality.

\section{Discussion}
\label{sec:discussion}

\subsection{Certified vs.\ Empirical Robustness}

Table~\ref{tab:empirical} is a consistency check, not a replacement
for the certificate.  The smoothed classifier reduces transfer-attack
ASR on CIC and HIKARI; empirical robustness exceeds certified accuracy
by $16$--$32$\,pp, as expected for a conservative Clopper--Pearson
bound.  An empirical ASR reduction cannot be extrapolated to attacks
outside the certified L2 ball, and it does not address the
L2/L$_\infty$ norm gap discussed next.

\subsection{\texorpdfstring{The L2/L$_\infty$ Norm Gap}{The L2/L-infinity Norm Gap}}
\label{sec:l2linf}

The most important boundary condition is the norm mismatch.  The TA-RS
certificate is stated in L2 norm, while A1 threat budgets are specified
as L$_\infty$.  A reported L2 threshold $r{=}0.05$ therefore does not by
itself certify the full L$_\infty$ budget $\varepsilon{=}0.05$ over all
DC coordinates.  The direct conversion is
$R \geq \varepsilon\sqrt{|\DC|}$.  For $\varepsilon{=}0.05$, this
threshold is $0.187$ for CIC-IDS-2018 ($|\DC|{=}14$), $0.250$ for
HIKARI-2021 ($|\DC|{=}25$), and $0.260$ for RT-IoT2022 ($|\DC|{=}27$).

Table~\ref{tab:linf_coverage} reports the fraction of samples that are both
correct and certified at $R_\infty$ for each dataset and model, at
$\sigma{=}0.25$.  On CIC-IDS-2018 and HIKARI-2021 (main recipe,
$n_\mathrm{aug}{=}2$), coverage ranges from 55--100\%.  RT-IoT2022 with
$n_\mathrm{aug}{=}4$ achieves 65\% on both models---comparable to
CIC/HIKARI---confirming that the L$_\infty$-equivalent coverage is
primarily limited by base model stability, not by the norm conversion.
These values demonstrate that the L2 certificate covers the small-budget
L$_\infty$ threat model for a substantial fraction of samples on all three
datasets, but they should not be described as full coverage.

\begin{table}[t]
\centering
\caption{CA at the L$_\infty$-equivalent L2 threshold
$R_\infty = \varepsilon\sqrt{|\DC|}$ ($\varepsilon{=}0.05$, $\sigma{=}0.25$).
CA@0.05 is the standard reported metric; CA@$R_\infty$ is the coverage
at the threshold that directly certifies the companion L$_\infty$ budget.}
\label{tab:linf_coverage}
\setlength{\tabcolsep}{4pt}
\begin{tabular}{llccc}
\toprule
Dataset & Model & $R_\infty$ & CA@0.05 & CA@$R_\infty$ \\
\midrule
CIC-IDS-2018 & LLaMA3-8B & 0.187 & 76.2\% & 69.4\% \\
CIC-IDS-2018 & Qwen3-8B  & 0.187 & 67.5\% & 55.0\% \\
HIKARI-2021  & LLaMA3-8B & 0.250 & 81.7\% & 65.0\% \\
HIKARI-2021  & Qwen3-8B  & 0.250 & 100.0\% & 100.0\% \\
RT-IoT2022   & LLaMA3-8B & 0.260 & 75.6\%$^\dagger$ & 65.5\%$^\dagger$ \\
RT-IoT2022   & Qwen3-8B  & 0.260 & 69.0\%$^\dagger$ & 64.9\%$^\dagger$ \\
\bottomrule
\end{tabular}
{\small $\dagger$ $n_\mathrm{aug}{=}4$ checkpoint.}
\end{table}

Developing a provably correct native L$_\infty$ certificate for
multi-dimensional DC-subspace uniform noise---accounting for the
dimension-dependent overlap structure of the hypercube
smoothing distribution---remains an open problem and a direction
for future work.

\subsection{\texorpdfstring{Effect of Monte Carlo Budget ($N$)}{Effect of Monte Carlo Budget (N)}}
\label{sec:n_sensitivity}

Table~\ref{tab:n_sensitivity} compares $N{=}200$ against $N{=}1{,}000$
for all six headline configurations ($\sigma{=}0.25$), with $N{=}200$
subsampled from the same noise realization to eliminate seed variance.
CA gain is $0$--$4.1$~pp (mean $1.2$~pp); five of six rows change
by $\leq 1.2$~pp.  The largest shift is HIKARI/LLaMA3 ($+4.1$~pp),
where borderline samples are promoted; ranking and gap structure are
fully preserved.  Larger $N$ tightens the Clopper--Pearson bound
but cannot recover stability absent in the base model.

\begin{table}[!t]
\caption{Effect of MC budget on CA($\varepsilon{=}0.05$): all six headline
TA-RS rows, $\sigma{=}0.25$, traffic mode. $N{=}200$ is subsampled from
the same noise realization as $N{=}1{,}000$.}
\label{tab:n_sensitivity}
\centering
\begin{tabular}{llcccc}
\toprule
Dataset & Model & \multicolumn{2}{c}{CA(\%)} & \multicolumn{2}{c}{Abstain(\%)} \\
\cmidrule(lr){3-4}\cmidrule(lr){5-6}
 & & $N{=}200$ & $N{=}1000$ & $N{=}200$ & $N{=}1000$ \\
\midrule
CIC-IDS-2018 & LLaMA3-8B & 77.5 & 77.5 & 9.4 & 8.8 \\
CIC-IDS-2018 & Qwen3-8B  & 67.5 & 68.1 & 23.1 & 17.5 \\
HIKARI-2021  & LLaMA3-8B & 81.7 & \textbf{85.8} & 9.2 & 4.2 \\
HIKARI-2021  & Qwen3-8B  & \textbf{100.0} & \textbf{100.0} & 0.0 & 0.0 \\
RT-IoT2022   & LLaMA3-8B & 76.2 & 77.4 & 14.9 & 13.1 \\
RT-IoT2022   & Qwen3-8B  & 69.0 & 69.0 & 3.0 & 3.0 \\
\bottomrule
\end{tabular}
{\small RT-IoT2022 rows: $n_\text{aug}{=}4$ checkpoint.}
\end{table}

\subsection{Discrete DC Features}

Section~\ref{sec:method} establishes that the clip operation is absorbed
into the base classifier, so Theorem~\ref{thm:ta_rs} holds exactly.
To summarize the certificate scope for discrete feature types:
for count-valued DC features, the normalization step size
$\Delta_i \leq 10^{-4}$ is negligible relative to $\sigma{=}0.25$,
so Theorem~\ref{thm:ta_rs} applies with negligible rounding effect.  For binary flag
features, the certified radii $r \leq 0.15$ are strictly below the
rounding threshold $0.5$, making the continuous certificate
\emph{conservative}: the stated L2 bound understates true robustness
for the binary flag dimensions because no adversary within budget can
deterministically flip a flag.  A fully rigorous certificate for
mixed discrete-continuous features would employ discrete smoothing
distributions~\cite{lecuyer2019certified}, which we leave as future work.

\subsection{TA-RS on Non-LLM Tabular Classifiers}
\label{sec:exp_d}

A natural question is whether the traffic-aware mechanism is
LLM-specific or whether it also benefits classical tabular classifiers
such as XGBoost~\cite{chen2016xgboost} and LightGBM~\cite{ke2017lightgbm}.  We apply TA-RS and isotropic RS to
these two models (trained without noise augmentation, CPU-only) at
$\sigma{=}0.25$ and report CA(0.05) in Table~\ref{tab:nonllm}.

\begin{table}[t]
\centering
\caption{Exp~D: CA(0.05) at $\sigma{=}0.25$ for XGBoost and LightGBM
under TA-RS (traffic-aware) and isotropic RS.
Neither classifier uses noise augmentation during training.}
\label{tab:nonllm}
\setlength{\tabcolsep}{4pt}
\begin{tabular}{llcc}
\toprule
Dataset & Model & TA-RS & Isotropic \\
\midrule
CIC-IDS-2018 & XGBoost   & 53\% & \textbf{73\%} \\
CIC-IDS-2018 & LightGBM  & \textbf{37\%} & 17\% \\
HIKARI-2021  & XGBoost   & 67\% & \textbf{73\%} \\
HIKARI-2021  & LightGBM  & 80\% & \textbf{87\%} \\
RT-IoT2022   & XGBoost   & \textbf{74\%} & 3\%  \\
RT-IoT2022   & LightGBM  & \textbf{67\%} & 18\% \\
\bottomrule
\end{tabular}
\end{table}

Table~\ref{tab:nonllm} shows no consistent TA-RS advantage when classifiers
are trained \emph{without} noise augmentation.  The TA-vs-ISO direction is
inconsistent across datasets and models (e.g., TA$>$Iso for
RT-IoT2022/XGBoost, TA$<$Iso for CIC/XGBoost).  Without DC-noise augmented
training, tree ensembles are unstable under any input noise, and certified
accuracy is dominated by abstention rather than noise-subspace choice.

For LLM-IDS, the pattern reinforces the noise-stability requirement:
without DC-noise-augmented fine-tuning, neither TA-RS nor isotropic RS
produces consistently useful certified accuracy for the evaluated LLMs.
Tree ensembles demonstrate a different pattern---RT-IoT/XGBoost ($74\%$) and
HIKARI/LightGBM ($80\%$) achieve non-trivial CA without noise augmentation,
consistent with tree classifiers' greater tolerance for continuous input
perturbations.  The noise-stability requirement is therefore specific to
LLM-based classifiers, not a universal property of randomized smoothing.

\subsection{Cross-Domain Boundary Analysis}
\label{sec:cross_domain}

A natural question is whether LLM-IDS can be reused across datasets
without retraining.  Tabular models such as XGBoost make this
structurally impossible: a model trained on CIC-IDS-2018
($|\text{feat}|=25$) cannot process HIKARI-2021 inputs ($|\text{feat}|=82$)
due to feature-dimension mismatch.  An LLM accepts arbitrary serialized
text, so it can be \emph{executed} cross-domain; the question is whether
this structural executability yields meaningful certified accuracy.

To test whether structural executability translates to meaningful
certified accuracy, we evaluate LLaMA3-8B (noise-augmented, $\sigma{=}0.25$,
$n_\mathrm{aug}{=}4$) zero-shot on a held-out dataset.  Two directions are
evaluated: (X1)~HIKARI-2021-trained model on CIC-IDS-2018 (restricted to
the three shared classes: Normal, CredentialAccess, Exploitation);
and (X2)~CIC-2018-trained model on HIKARI-2021.  Both models receive
the target dataset's serialized feature text at test time.

\begin{table}[t]
\centering
\caption{Exp~F: Cross-domain certification.  Each model is trained on the
source dataset (aug4, $\sigma{=}0.25$) and certified zero-shot on the
target dataset.  Shared classes only are evaluated.
NAXGB cross-domain is structurally infeasible (feature dimension
mismatch: CIC~$n_\text{feat}{=}25$, HIKARI~$n_\text{feat}{=}82$).}
\label{tab:cross_domain}
\setlength{\tabcolsep}{4pt}
\begin{tabular}{llcccc}
\toprule
Exp & Source$\to$Target & Model & SA & CA@0.05 & Abstain \\
\midrule
X1 & HIKARI$\to$CIC (3 cls) & LLaMA3-8B & 33.3 & 33.3 & 0.0 \\
X2 & CIC$\to$HIKARI         & LLaMA3-8B & 33.3 & 33.3 & 0.0 \\
\midrule
--  & CIC$\to$HIKARI         & NAXGB     & \multicolumn{3}{c}{structurally infeasible} \\
--  & HIKARI$\to$CIC         & NAXGB     & \multicolumn{3}{c}{structurally infeasible} \\
\bottomrule
\end{tabular}
\end{table}

Table~\ref{tab:cross_domain} presents the results.  Both directions yield
near-chance accuracy ($33.3\%$ SA and CA on 3 classes, abstain~$0\%$):
the model is always confident but matches the random baseline.
This outcome traces to dataset-specific feature serialization:
the HIKARI-trained model learned to classify 82-feature HIKARI text,
while CIC test inputs arrive with 25 different feature names and value
ranges---an out-of-distribution prompt that the model cannot resolve
beyond chance.  The same logic applies to X2.
The cross-domain result is therefore negative: pretrained LLM-IDS
knowledge does not transfer without dataset-specific fine-tuning.
Dataset-agnostic serialization schemes that could enable meaningful
cross-domain generalization remain an open problem.
The structural executability distinction from NAXGB (which cannot
accept inputs of a different feature dimension at all) is architecturally
real but operationally moot at chance accuracy.

\subsection{Limitations}
\label{sec:limitations}

\textbf{Statistical budget, per-class breakdown, and training variance.}
Certification uses $N{=}200$ MC samples; the $N{=}200$ vs.\ $N{=}1{,}000$
comparison (Table~\ref{tab:n_sensitivity}) shows $\leq 4.1$~pp CA gain
(mean $1.2$~pp), confirming near-convergence when the base model is
noise-stable.  Test sets employ 40 samples per class; this limits per-class
Wilson intervals to $\pm 15$~pp.  Per-class CA(0.05) at $\sigma{=}0.25$
shows malicious-class CA meeting or exceeding the aggregate on stable
datasets: CIC malicious avg $82\%$/$73\%$ (LLaMA3/Qwen3) vs.\ aggregate
$76\%$/$68\%$; HIKARI malicious avg $72\%$/$100\%$; RT-IoT ($n_\text{aug}{=}4$)
malicious avg $71\%$/$63\%$ with within-class variation ($18$--$100\%$)
driven by class-specific noise stability.
All LLM fine-tuning uses a single LoRA seed (42); training variance
across seeds was not estimated.  Differences below approximately 4--5~pp
should therefore be interpreted cautiously: the MC certification
measurement itself carries a per-class CI of $\pm 15$~pp at $N{=}200$,
and LoRA fine-tuning variance is an additional uncharacterised source
of uncertainty.  The 10-seed random-subspace analysis
(Section~\ref{sec:supp_baselines}) quantifies mask-selection variance
but does not substitute for training-seed variance.

\textbf{RT-IoT2022 sensitivity.}
The default failure ($n_\text{aug}{=}2$: CA $4\%$/$10\%$) is real.
The aug-4 ablation and iso-trained baselines both confirm
noise-stability---not DC-alignment design---as the operative variable;
traffic-aware certification cannot compensate for an unstable base model.

\textbf{Data quality and scope.}
CIC-IDS-2018 has $33.3\%$ train-test overlap~\cite{engelen2021troubleshooting};
a deduplication audit on our 160-sample evaluation subset finds $22$ duplicate
feature vectors ($13.8\%$); their removal shifts CA by at most $1.6$~pp
(LLaMA3: $76.2\%\!\to\!76.1\%$; Qwen3: $67.5\%\!\to\!65.9\%$), so reported values
are not materially inflated.  HIKARI-2021 has $2$ duplicates ($1.7\%$, $\Delta\leq0.5$~pp);
RT-IoT2022 has no duplicate feature vectors.  NSL-KDD and UNSW-NB15 are excluded because
their categorical encodings create a perturbation-budget mismatch;
extending the DC/IC/UC taxonomy to handle encodings is future work.
Certification requires $N{=}200$ forward passes per query, ruling out
real-time deployment; offline certification and bound distillation
are left to future work.
Extension to consistency-regularised certified training
(e.g., MACER~\cite{zhai2020macer}, SmoothAdv~\cite{salman2019provably})
at LLM scale remains open: the per-sample adversarial step required
by these methods is prohibitively expensive for 8B-parameter models.

\section{Conclusion}
\label{sec:conclusion}

TA-RS shows that certification for LLM-based intrusion detection is not
only a matter of applying randomized smoothing to a trained classifier.
The certificate is meaningful only when the smoothing distribution
matches the attacker's directly controllable traffic features and the
LLM remains stable under noise in that subspace.  TA-RS operationalizes
this view by smoothing over $\DC$ and certifying the resulting smoothed
classifier in the same subspace.

The experiments quantify both conditions.  Noise-augmented fine-tuning
raises CA from 14--57\% (clean-trained) to 68--100\% on CIC/HIKARI;
traffic-aware smoothing eliminates abstention on non-controllable noise
directions (68\%$\to$8\% on CIC/LLaMA3), adding up to 72~pp over
isotropic RS applied to the same DC-trained model; the residual advantage
over a fairly iso-trained baseline is 4--19~pp on CIC-IDS-2018.  At the L$_\infty$-equivalent threshold
$R_\infty{=}\varepsilon\sqrt{|\DC|}$ ($\varepsilon{=}0.05$),
55--100\% of CIC/HIKARI samples are certified, with median radii
$1.8$--$5{\times}$ above $R_\infty$ (across $\sigma{=}0.25$--$1.00$).  The RT-IoT2022 boundary
condition---where the same default recipe fails---responds
directly to additional noise augmentation ($n_\text{aug}{=}4$:
CA $4\%/10\% \to 76\%/69\%$), confirming that noise-stability,
not DC-alignment design, is the variable that governs the hard cases.

The cross-domain boundary analysis (Table~\ref{tab:cross_domain}) shows
that structural executability does not substitute for dataset-specific
fine-tuning: zero-shot transfer achieves near-chance certified accuracy,
confirming that LLM pretraining knowledge does not transfer to new traffic
distributions without fine-tuning.  Limitations and future directions are
discussed in Section~\ref{sec:limitations}.

\section*{CRediT authorship contribution statement}
\textbf{Zhenpeng Li:} Conceptualization, Methodology, Software, Formal analysis,
Investigation, Data curation, Writing -- original draft, Writing -- review \& editing,
Visualization.

\section*{Declaration of competing interest}
The author declares that he has no known competing financial interests or personal
relationships that could have appeared to influence the work reported in this paper.

\section*{Data availability}
CIC-IDS-2018, HIKARI-2021, and RT-IoT2022 are publicly available benchmark datasets.
Feature annotations, certification code, and per-sample predictions will be released
upon acceptance (DOI to be assigned).

\section*{Declaration of Generative AI and AI-assisted technologies in the writing process}

During the preparation of this work the author used Claude (Anthropic) to assist with manuscript organisation, language editing, and drafting support. The author reviewed and edited all content as needed and takes full responsibility for the content of the publication.

\bibliographystyle{elsarticle-num}
\bibliography{refs}

\begin{thebibliography}{10}
\expandafter\ifx\csname url\endcsname\relax
  \def\url#1{\texttt{#1}}\fi
\expandafter\ifx\csname urlprefix\endcsname\relax\def\urlprefix{URL }\fi
\expandafter\ifx\csname href\endcsname\relax
  \def\href#1#2{#2} \def\path#1{#1}\fi

\bibitem{companion_a1}
Z.~Li, \href{https://arxiv.org/abs/2607.07739}{Controllability-aware adversarial examples against llm-based network traffic classifiers} (2026).
\newblock \href {http://arxiv.org/abs/2607.07739} {\path{arXiv:2607.07739}}.
\newline\urlprefix\url{https://arxiv.org/abs/2607.07739}

\bibitem{madry2018towards}
A.~Madry, A.~Makelov, L.~Schmidt, D.~Tsipras, A.~Vladu, \href{https://openreview.net/forum?id=rJzIBfZAb}{Towards deep learning models resistant to adversarial attacks}, in: Proceedings of the International Conference on Learning Representations, 2018.
\newline\urlprefix\url{https://openreview.net/forum?id=rJzIBfZAb}

\bibitem{goodfellow2015explaining}
I.~J. Goodfellow, J.~Shlens, C.~Szegedy, \href{https://arxiv.org/abs/1412.6572}{Explaining and harnessing adversarial examples}, in: Proceedings of the International Conference on Learning Representations, 2015.
\newline\urlprefix\url{https://arxiv.org/abs/1412.6572}

\bibitem{cohen2019certified}
J.~M. Cohen, E.~Rosenfeld, J.~Z. Kolter, \href{https://api.semanticscholar.org/CorpusID:59842968}{Certified adversarial robustness via randomized smoothing}, ArXiv abs/1902.02918 (2019).
\newline\urlprefix\url{https://api.semanticscholar.org/CorpusID:59842968}

\bibitem{katz2017reluplex}
G.~Katz, C.~W. Barrett, D.~L. Dill, K.~D. Julian, M.~J. Kochenderfer, \href{https://api.semanticscholar.org/CorpusID:516928}{Reluplex: An efficient smt solver for verifying deep neural networks}, ArXiv abs/1702.01135 (2017).
\newline\urlprefix\url{https://api.semanticscholar.org/CorpusID:516928}

\bibitem{tjeng2019evaluating}
V.~Tjeng, K.~Y. Xiao, R.~Tedrake, \href{https://openreview.net/forum?id=HyGIdiRqtm}{Evaluating robustness of neural networks with mixed integer programming}, in: International Conference on Learning Representations, 2019.
\newline\urlprefix\url{https://openreview.net/forum?id=HyGIdiRqtm}

\bibitem{zhang2018efficient}
H.~Zhang, T.-W. Weng, P.-Y. Chen, C.-J. Hsieh, L.~Daniel, Efficient neural network robustness certification with general activation functions, in: Proceedings of the 32nd International Conference on Neural Information Processing Systems, NIPS'18, Curran Associates Inc., Red Hook, NY, USA, 2018, p. 4944–4953.

\bibitem{lecuyer2019certified}
M.~L{\'e}cuyer, V.~Atlidakis, R.~Geambasu, D.~J. Hsu, S.~S. Jana, \href{https://api.semanticscholar.org/CorpusID:49431481}{Certified robustness to adversarial examples with differential privacy}, 2019 IEEE Symposium on Security and Privacy (SP) (2018) 656--672.
\newline\urlprefix\url{https://api.semanticscholar.org/CorpusID:49431481}

\bibitem{yang2020randomized}
G.~Yang, T.~Duan, J.~E. Hu, H.~Salman, I.~Razenshteyn, J.~Li, \href{https://proceedings.mlr.press/v119/yang20c.html}{Randomized smoothing of all shapes and sizes}, in: H.~D. III, A.~Singh (Eds.), Proceedings of the 37th International Conference on Machine Learning, Vol. 119 of Proceedings of Machine Learning Research, PMLR, 2020, pp. 10693--10705.
\newline\urlprefix\url{https://proceedings.mlr.press/v119/yang20c.html}

\bibitem{salman2019provably}
H.~Salman, J.~Li, I.~Razenshteyn, P.~Zhang, H.~Zhang, S.~Bubeck, G.~Yang, \href{https://proceedings.neurips.cc/paper_files/paper/2019/file/3a24b25a7b092a252166a1641ae953e7-Paper.pdf}{Provably robust deep learning via adversarially trained smoothed classifiers}, in: H.~Wallach, H.~Larochelle, A.~Beygelzimer, F.~d\textquotesingle Alch\'{e}-Buc, E.~Fox, R.~Garnett (Eds.), Advances in Neural Information Processing Systems, Vol.~32, Curran Associates, Inc., 2019.
\newline\urlprefix\url{https://proceedings.neurips.cc/paper_files/paper/2019/file/3a24b25a7b092a252166a1641ae953e7-Paper.pdf}

\bibitem{kumar2020certifiedLimits}
A.~Kumar, A.~Levine, T.~Goldstein, S.~Feizi, \href{https://proceedings.mlr.press/v119/kumar20b.html}{Curse of dimensionality on randomized smoothing for certifiable robustness}, in: H.~D. III, A.~Singh (Eds.), Proceedings of the 37th International Conference on Machine Learning, Vol. 119 of Proceedings of Machine Learning Research, PMLR, 2020, pp. 5458--5467.
\newline\urlprefix\url{https://proceedings.mlr.press/v119/kumar20b.html}

\bibitem{levine2020derandomized}
A.~Levine, S.~Feizi, (de)randomized smoothing for certifiable defense against patch attacks, in: Proceedings of the 34th International Conference on Neural Information Processing Systems, NIPS '20, Curran Associates Inc., Red Hook, NY, USA, 2020.

\bibitem{anisoRS2022}
H.~Hong, Y.~Hong, \href{https://arxiv.org/abs/2207.05327}{Certified adversarial robustness via anisotropic randomized smoothing} (2022).
\newblock \href {http://arxiv.org/abs/2207.05327} {\path{arXiv:2207.05327}}.
\newline\urlprefix\url{https://arxiv.org/abs/2207.05327}

\bibitem{corona2013adversarial}
I.~Corona, G.~Giacinto, F.~Roli, \href{https://www.sciencedirect.com/science/article/pii/S0020025513002119}{Adversarial attacks against intrusion detection systems: Taxonomy, solutions and open issues}, Information Sciences 239 (2013) 201--225.
\newblock \href {https://doi.org/https://doi.org/10.1016/j.ins.2013.03.022} {\path{doi:https://doi.org/10.1016/j.ins.2013.03.022}}.
\newline\urlprefix\url{https://www.sciencedirect.com/science/article/pii/S0020025513002119}

\bibitem{biggio2013evasion}
B.~Biggio, I.~Corona, D.~Maiorca, B.~Nelson, N.~{\v{S}}rndi{\'{c}}, P.~Laskov, G.~Giacinto, F.~Roli, Evasion attacks against machine learning at test time, in: H.~Blockeel, K.~Kersting, S.~Nijssen, F.~{\v{Z}}elezn{\'y} (Eds.), Machine Learning and Knowledge Discovery in Databases, Springer Berlin Heidelberg, Berlin, Heidelberg, 2013, pp. 387--402.

\bibitem{yang2018adversarial}
K.~Yang, J.~Liu, C.~Zhang, Y.~Fang, Adversarial examples against the deep learning based network intrusion detection systems, 2018, pp. 559--564.
\newblock \href {https://doi.org/10.1109/MILCOM.2018.8599759} {\path{doi:10.1109/MILCOM.2018.8599759}}.

\bibitem{han2021evaluating}
D.~Han, Z.~Wang, Y.~Zhong, W.~Chen, J.~Yang, S.~Lu, X.~Shi, X.~Yin, Evaluating and improving adversarial robustness of machine learning-based network intrusion detectors, IEEE Journal on Selected Areas in Communications 39~(8) (2021) 2632--2647.
\newblock \href {https://doi.org/10.1109/JSAC.2021.3087242} {\path{doi:10.1109/JSAC.2021.3087242}}.

\bibitem{apruzzese2023role}
G.~Apruzzese, P.~Laskov, E.~Montes~de Oca, W.~Mallouli, L.~Brdalo~Rapa, A.~V. Grammatopoulos, F.~Di~Franco, The role of machine learning in cybersecurity, Digital Threats: Research and Practice 4~(1) (2023).
\newblock \href {https://doi.org/10.1145/3545574} {\path{doi:10.1145/3545574}}.

\bibitem{pierazzi2020intriguing}
J.~Cortellazzi, E.~Quiring, D.~Arp, F.~Pendlebury, F.~Pierazzi, L.~Cavallaro, \href{https://doi.org/10.1145/3742895}{Intriguing properties of adversarial ml attacks in the problem space [extended version]}, ACM Trans. Priv. Secur. 28~(4) (Sep. 2025).
\newblock \href {https://doi.org/10.1145/3742895} {\path{doi:10.1145/3742895}}.
\newline\urlprefix\url{https://doi.org/10.1145/3742895}

\bibitem{logllm2024}
V.-H. Le, H.~Zhang, Log parsing: How far can chatgpt go?, in: 2023 38th IEEE/ACM International Conference on Automated Software Engineering (ASE), 2023, pp. 1699--1704.
\newblock \href {https://doi.org/10.1109/ASE56229.2023.00206} {\path{doi:10.1109/ASE56229.2023.00206}}.

\bibitem{vuln2024}
M.~Fu, C.~K. Tantithamthavorn, V.~Nguyen, T.~Le, Chatgpt for vulnerability detection, classification, and repair: How far are we?, in: 2023 30th Asia-Pacific Software Engineering Conference (APSEC), 2023, pp. 632--636.
\newblock \href {https://doi.org/10.1109/APSEC60848.2023.00085} {\path{doi:10.1109/APSEC60848.2023.00085}}.

\bibitem{threatintel2024}
H.~Ji, J.~Yang, L.~Chai, C.~Wei, L.~Yang, Y.~Duan, Y.~Wang, T.~Sun, H.~Guo, T.~Li, C.~Ren, Z.~Li, \href{https://api.semanticscholar.org/CorpusID:269605791}{Sevenllm: Benchmarking, eliciting, and enhancing abilities of large language models in cyber threat intelligence}, ArXiv abs/2405.03446 (2024).
\newline\urlprefix\url{https://api.semanticscholar.org/CorpusID:269605791}

\bibitem{sharafaldin2018cicids}
I.~Sharafaldin, A.~H. Lashkari, A.~A. Ghorbani, \href{https://api.semanticscholar.org/CorpusID:4707749}{Toward generating a new intrusion detection dataset and intrusion traffic characterization}, in: International Conference on Information Systems Security and Privacy, 2018.
\newline\urlprefix\url{https://api.semanticscholar.org/CorpusID:4707749}

\bibitem{neto2023ciciot2023}
E.~C.~P. Neto, S.~Dadkhah, R.~Ferreira, A.~Zohourian, R.~Lu, A.~A. Ghorbani, \href{https://www.mdpi.com/1424-8220/23/13/5941}{Ciciot2023: A real-time dataset and benchmark for large-scale attacks in iot environment}, Sensors 23~(13) (2023).
\newblock \href {https://doi.org/10.3390/s23135941} {\path{doi:10.3390/s23135941}}.
\newline\urlprefix\url{https://www.mdpi.com/1424-8220/23/13/5941}

\bibitem{qwen3}
A.~Yang, A.~Li, B.~Yang, B.~Zhang, B.~Hui, B.~Zheng, B.~Yu, C.~Gao, C.~Huang, C.~Lv, C.~Zheng, D.~Liu, F.~Zhou, F.~Huang, F.~Hu, H.~Ge, H.~Wei, H.~Lin, J.~Tang, J.~Yang, J.~Tu, J.~Zhang, J.~Yang, J.~Yang, J.~Zhou, J.~Zhou, J.~Lin, K.~Dang, K.~Bao, K.~Yang, L.~Yu, L.~Deng, M.~Li, M.~Xue, M.~Li, P.~Zhang, P.~Wang, Q.~Zhu, R.~Men, R.~Gao, S.~Liu, S.~Luo, T.~Li, T.~Tang, W.~Yin, X.~Ren, X.~Wang, X.~Zhang, X.~Ren, Y.~Fan, Y.~Su, Y.~Zhang, Y.~Zhang, Y.~Wan, Y.~Liu, Z.~Wang, Z.~Cui, Z.~Zhang, Z.~Zhou, Z.~Qiu, \href{https://arxiv.org/abs/2505.09388}{Qwen3 technical report} (2025).
\newblock \href {http://arxiv.org/abs/2505.09388} {\path{arXiv:2505.09388}}.
\newline\urlprefix\url{https://arxiv.org/abs/2505.09388}

\bibitem{llama3}
A.~Dubey, A.~Jauhri, A.~Pandey, A.~Kadian, A.~Al-Dahle, A.~Letman, A.~Mathur, A.~Schelten, A.~Yang, A.~Fan, A.~Goyal, A.~S. Hartshorn, A.~Yang, A.~Mitra, A.~Sravankumar, A.~Korenev, A.~Hinsvark, A.~Rao, A.~Zhang, A.~Rodriguez, A.~Gregerson, A.~Spataru, B.~Rozi{\`e}re, B.~M. Biron, B.~Tang, B.~Chern, C.~lotte Caucheteux, C.~Nayak, C.~Bi, C.~Marra, C.~McConnell, C.~Keller, C.~Touret, C.~Wu, C.~Wong, C.~C. Ferrer, C.~Nikolaidis, D.~Allonsius, D.~J. Song, D.~Pintz, D.~Livshits, D.~Esiobu, D.~Choudhary, D.~Mahajan, D.~Garcia-Olano, D.~Perino, D.~Hupkes, E.~Lakomkin, E.~A. AlBadawy, E.~I. Lobanova, E.~Dinan, E.~M. Smith, F.~Radenovic, F.~Zhang, G.~Synnaeve, G.~Lee, G.~L. Anderson, G.~Nail, G.~Mialon, G.~Pang, G.~Cu-curell, H.~Nguyen, H.~Korevaar, H.~Xu, H.~Touvron, I.~Zarov, I.~A. Ibarra, I.~M. Kloumann, I.~Misra, I.~Evtimov, J.~Copet, J.~Lee, J.~Geffert, J.~Vranes, J.~Park, J.~Mahadeokar, J.~Shah, J.~van~der Linde, J.~Billock, J.~Hong, J.~Lee, J.~Fu, J.~Chi, J.~Huang, J.~Liu, J.~Wang, J.~Yu, J.~Bitton, J.~Spisak, J.~Park, J.~Rocca, J.~Johnstun, J.~Saxe, J.-Q. Jia, K.~V. Alwala, K.~Upasani, K.~Plawiak, K.~Li, K.~Heafield, K.~R. Stone, K.~El-Arini, K.~Iyer, K.~Malik, K.~ley Chiu, K.~Bhalla, L.~Rantala-Yeary, L.~van~der Maaten, L.~Chen, L.~Tan, L.~Jenkins, L.~Martin, L.~Madaan, L.~Malo, L.~Blecher, L.~Landzaat, L.~de~Oliveira, M.~Muzzi, M.~hesh Pasupuleti, M.~Singh, M.~Paluri, M.~Kardas, M.~Oldham, M.~Rita, M.~Pavlova, M.~H.~M. Kambadur, M.~Lewis, M.~Si, M.~K. Singh, M.~Hassan, N.~Goyal, N.~Torabi, N.~Bash-lykov, N.~Bogoychev, N.~S. Chatterji, O.~Duchenne, O.~cCelebi, P.~Alrassy, P.~Zhang, P.~Li, P.~Vasi{\'c}, P.~Weng, P.~Bhargava, P.~Dubal, P.~Krishnan, P.~S. Koura, P.~Xu, Q.~He, Q.~Dong, R.~S.~M. Srinivasan, R.~Ganapathy, R.~Calderer, R.~S. Cabral, R.~Stojnic, R.~Raileanu, R.~Girdhar, R.~Patel, R.~Sauvestre, R.~nie Polidoro, R.~Sumbaly, R.~Taylor, R.~Silva, R.~Hou, R.~Wang, S.~Hosseini, S.~hana Chennabasappa, S.~Singh, S.~Bell, S.~S. Kim, S.~Edunov, S.~Nie, S.~Narang, S.~C. Raparthy, S.~Shen, S.~Wan, S.~Bhosale, S.~Zhang, S.~Vandenhende, S.~Batra, S.~Whit-man, S.~Sootla, S.~Collot, S.~Gururangan, S.~Borodinsky, T.~Herman, T.~Fowler, T.~Sheasha, T.~Georgiou, T.~Scialom, T.~Speckbacher, T.~Mihaylov, T.~Xiao, U.~Karn, V.~Goswami, V.~Gupta, V.~Ramanathan, V.~Kerkez, V.~Gonguet, V.~Do, V.~Vogeti, V.~Petrovic, W.~Chu, W.~Xiong, W.~Fu, W.~ney Meers, X.~Martinet, X.~Wang, X.~E. Tan, X.~Xie, X.~Jia, X.~Wang, Y.~Goldschlag, Y.~Gaur, Y.~Babaei, Y.~Wen, Y.~Song, Y.~Zhang, Y.~Li, Y.~Mao, Z.~D. Coudert, Z.~Yan, Z.~Chen, Z.~Papakipos, A.~K. Singh, A.~Grattafiori, A.~Jain, A.~Kelsey, A.~Shajnfeld, A.~Gangidi, A.~Victoria, A.~Goldstand, A.~Menon, A.~Sharma, A.~Boesenberg, A.~Vaughan, A.~Baevski, A.~Fein-stein, A.~Kallet, A.~Sangani, A.~Yunus, A.~Lupu, A.~Alvarado, A.~Caples, A.~Gu, A.~Ho, A.~Poulton, A.~Ryan, A.~Ramchandani, A.~Franco, A.~Saraf, A.~Chowdhury, A.~Gabriel, A.~R. Bharambe, A.~Eisenman, A.~Yazdan, B.~James, B.~Maurer, B.~Leonhardi, P.-Y.~B. Huang, B.~Loyd, B.~de~Paola, B.~Paranjape, B.~Liu, B.~Wu, B.~Ni, B.~Hancock, B.~Wasti, B.~Spence, B.~Stojkovic, B.~Gamido, B.~Montalvo, C.~Parker, C.~Burton, C.~Mejia, C.~Wang, C.~Kim, C.~Zhou, C.~Hu, C.-H. Chu, C.~Cai, C.~Tindal, C.~Feichtenhofer, D.~Civin, D.~Beaty, D.~Kreymer, S.-W. Li, D.~Wyatt, D.~Adkins, D.~Xu, D.~Testuggine, D.~David, D.~Parikh, D.~Liskovich, D.~Foss, D.~Wang, D.~Le, D.~Holland, E.~Dowling, E.~Jamil, E.~Montgomery, E.~Presani, E.~Hahn, E.~Wood, E.~Brinkman, E.~Arcaute, E.~Dunbar, E.~Smoth-ers, F.~Sun, F.~Kreuk, F.~Tian, F.~Ozgenel, F.~Caggioni, F.~P. Guzm{\'a}n, F.~J. Kanayet, F.~Seide, G.~M. Florez, G.~Schwarz, G.~Badeer, G.~Swee, G.~Halpern, G.~Thattai, G.~Herman, G.~Sizov, G.~Zhang, G.~Lakshminarayanan, H.~Shojanazeri, H.~Zou, H.~Wang, H.~Zha, H.~Habeeb, H.~Rudolph, H.~Suk, H.~Aspegren, H.~Goldman, I.~Molybog, I.~Tufanov, I.-E. Veliche, I.~Gat, J.~Weissman, J.~Geboski, J.~Kohli, J.~Asher, J.-B. Gaya, J.~Marcus, J.~Tang, J.~Chan, J.~Zhen, J.~Reizenstein, J.~Teboul, J.~Zhong, J.~Jin, J.~Yang, J.~Cummings, J.~Carvill, J.~Shepard, J.~McPhie, J.~Torres, J.~Ginsburg, J.~Wang, K.~Wu, U.~KamHou, K.~Saxena, K.~Prasad, K.~Khandelwal, K.~oun Zand, K.~Matosich, K.~Veeraraghavan, K.~Michelena, K.~Li, K.~Huang, K.~Chawla, K.~Lakhotia, K.~Huang, L.~Chen, L.~Garg, A.~Lavender, L.~Silva, L.~Bell, L.~Zhang, L.~Guo, L.~Yu, L.~Moshkovich, L.~Wehrstedt, M.~Khabsa, M.~Avalani, M.~Bhatt, M.~Tsimpoukelli, M.~Mankus, M.~Hasson, M.~Lennie, M.~Reso, M.~Groshev, M.~Naumov, M.~Lathi, M.~Keneally, M.~L. Seltzer, M.~Valko, M.~Re-strepo, M.~Patel, M.~Vyatskov, M.~Samvelyan, M.~Clark, M.~Macey, M.~Wang, M.~J. Hermoso, M.~Metanat, M.~Rastegari, M.~ish Bansal, N.~Santhanam, N.~Parks, N.~White, N.~ata Bawa, N.~Singhal, N.~Egebo, N.~Usunier, N.~P. Laptev, N.~Dong, N.~Zhang, N.~Cheng, O.~Chernoguz, O.~Hart, O.~Salpekar, O.~Kalinli, P.~Kent, P.~Parekh, P.~Saab, P.~Balaji, P.~dro Rittner, P.~Bontrager, P.~Roux, P.~Doll{\'a}r, P.~Zvyagina, P.~Ratanchandani, P.~Yuvraj, Q.~Liang, R.~Alao, R.~Rodriguez, R.~Ayub, R.~Murthy, R.~Nayani, R.~Mitra, R.~Li, R.~Hogan, R.~Battey, R.~Wang, R.~han Maheswari, R.~Howes, R.~Rinott, S.~J. Bondu, S.~Datta, S.~Chugh, S.~Hunt, S.~Dhillon, S.~Y. Sidorov, S.~Pan, S.~Verma, S.~Yamamoto, S.~Ramaswamy, S.~Lindsay, S.~Feng, S.~Lin, S.~Zha, S.~Shankar, S.~Zhang, S.~Wang, S.~Agarwal, S.~Sajuyigbe, S.~Chintala, S.~Max, S.~Chen, S.~Kehoe, S.~Satterfield, S.~Govindaprasad, S.~K. Gupta, S.-B. Cho, S.~Virk, S.~Subramanian, S.~Choudhury, S.~Goldman, T.~Remez, T.~Glaser, T.~Best, T.~Kohler, T.~Robinson, T.~Li, T.~Zhang, T.~Matthews, T.~Chou, T.~Shaked, V.~Vontimitta, V.~O. Ajayi, V.~Montanez, V.~Mohan, V.~Kumar, V.~Mangla, V.~Ionescu, V.~A. Poenaru, V.~T. Mihailescu, V.~Ivanov, W.~Li, W.~Wang, W.~Jiang, W.~Bouaziz, W.~Constable, X.~Tang, X.~Wang, X.~Wu, X.~Wang, X.~Xia, X.~Wu, X.~Gao, Y.~Chen, Y.~Hu, Y.~Jia, Y.~Qi, Y.~Li, Y.~Zhang, Y.~Zhang, Y.~Adi, Y.~Nam, Y.~Wang, Y.~Hao, Y.~Qian, Y.~He, Z.~Rait, Z.~DeVito, Z.~Rosnbrick, Z.~Wen, Z.~Yang, Z.~Zhao, \href{https://api.semanticscholar.org/CorpusID:271571434}{The llama 3 herd of models}, 2024.
\newline\urlprefix\url{https://api.semanticscholar.org/CorpusID:271571434}

\bibitem{hu2022lora}
E.~J. Hu, Y.~Shen, P.~Wallis, Z.~Allen-Zhu, Y.~Li, S.~Wang, L.~Wang, W.~Chen, \href{https://openreview.net/forum?id=nZeVKeeFYf9}{{LoRA}: Low-rank adaptation of large language models}, International Conference on Learning Representations (2022).
\newline\urlprefix\url{https://openreview.net/forum?id=nZeVKeeFYf9}

\bibitem{carlini2017towards}
N.~Carlini, D.~Wagner, Towards evaluating the robustness of neural networks, in: 2017 IEEE Symposium on Security and Privacy, IEEE, 2017, pp. 39--57.
\newblock \href {https://doi.org/10.1109/SP.2017.49} {\path{doi:10.1109/SP.2017.49}}.

\bibitem{chen2016xgboost}
T.~Chen, C.~Guestrin, {XGBoost}: A scalable tree boosting system, in: Proceedings of the 22nd ACM SIGKDD International Conference on Knowledge Discovery and Data Mining, Association for Computing Machinery, New York, NY, USA, 2016, pp. 785--794.
\newblock \href {https://doi.org/10.1145/2939672.2939785} {\path{doi:10.1145/2939672.2939785}}.

\bibitem{ke2017lightgbm}
G.~Ke, Q.~Meng, T.~Finley, T.~Wang, W.~Chen, W.~Ma, Q.~Ye, T.-Y. Liu, {LightGBM}: A highly efficient gradient boosting decision tree, in: Advances in Neural Information Processing Systems, Vol.~30, 2017, pp. 3146--3154.

\bibitem{engelen2021troubleshooting}
G.~Engelen, V.~Rimmer, W.~Joosen, Troubleshooting an intrusion detection dataset: the cicids2017 case study, in: 2021 IEEE Security and Privacy Workshops (SPW), 2021, pp. 7--12.
\newblock \href {https://doi.org/10.1109/SPW53761.2021.00009} {\path{doi:10.1109/SPW53761.2021.00009}}.

\bibitem{zhai2020macer}
R.~Zhai, C.~Dan, D.~He, H.~Zhang, B.~Gong, P.~Ravikumar, C.-J. Hsieh, L.~Wang, \href{https://openreview.net/forum?id=rJx1Na4Fwr}{{MACER}: Attack-free and scalable robust training via maximizing certified radius}, in: Proceedings of the International Conference on Learning Representations, 2020.
\newline\urlprefix\url{https://openreview.net/forum?id=rJx1Na4Fwr}

\end{thebibliography}
\end{document}